\documentclass[journal]{IEEEtran}
\IEEEoverridecommandlockouts
\usepackage{makecell}
\usepackage{cite}
\usepackage{amsmath,amssymb,amsfonts}

\usepackage{algorithmic}
\usepackage{graphicx}
\usepackage{textcomp}
\newtheorem{theorem}{Theorem}
\newtheorem{lemma}{Lemma}
\newtheorem{remark}{Remark}
\usepackage{bbm}
\usepackage{bm}
\newtheorem{proof}{Proof}

\usepackage{bbm}
\usepackage{pifont}
\usepackage[table]{xcolor}

\begin{document}
\setlength{\textfloatsep}{0.11cm}
\setlength{\baselineskip}{0.42cm}
\setlength{\abovecaptionskip}{1mm}

\title{{Capacity and IAPR Analysis for MIMO Faster-than-Nyquist Signaling with High Acceleration Rate}}
\author{Zichao~Zhang,~\IEEEmembership{Student Member,~IEEE,}
		Melda~Yuksel,~\IEEEmembership{Senior Member,~IEEE,}
  Gokhan M. Guvensen,
Halim~Yanikomeroglu,~\IEEEmembership{Fellow,~IEEE}
\thanks{This work was funded in part by the Scientific and Technological Research Council of Turkey, TUBITAK, under grant 122E248, and in part by a Discovery Grant awarded by the Natural Sciences and Engineering Research Council of Canada (NSERC).}
\thanks{Z. Zhang and H. Yanikomeroglu are with the Department of Systems and Computer Engineering at Carleton University, Ottawa, ON, K1S 5B6, Canada e-mail:	zichaozhang@cmail.carleton.ca, halim@sce.carleton.ca.}
\thanks{M. Yuksel and G. Guvensen are with the Department of Electrical and Electronics Engineering, Middle East Technical University, Ankara, 06800, Turkey, e-mail: {ymelda, guvensen}@metu.edu.tr.}
}

\maketitle

\begin{abstract}
 Faster-than-Nyquist (FTN) signaling is a non-orthogonal transmission technique offering a promising solution for future generations of communications.  
{This paper
studies the capacity of FTN signaling in multiple-input multiple-output
(MIMO) channels for high acceleration factors. In our previous study \cite{zhang2022faster}, we found the capacity for MIMO FTN channels if the acceleration factor is larger than a certain threshold, which depends on the bandwidth of the pulse shape used.   In this paper we extend the capacity analysis to acceleration factors smaller than this mentioned threshold.} In addition to capacity, we conduct {instantaneous-to-average power ratio (IAPR)} analysis and simulation for MIMO FTN for varying acceleration factors for both Gaussian and QPSK symbol sets. Our analysis reveals important insights about transmission power and received signal-to-noise ratio (SNR) variation in FTN. As the acceleration factor approaches 0, if the transmission power is fixed, the received SNR diminishes, or if the received SNR is fixed, {IAPR} at the transmitter explodes. 
\end{abstract}

\begin{IEEEkeywords}
Channel capacity, faster-than-Nyquist, multiple-input multiple-output, {instantaneous} to average power ratio.
\end{IEEEkeywords}

\section{Introduction}

Faster-than-Nyquist (FTN) transmission emerges as a groundbreaking technology in modern communication systems, poised to redefine data transmission capabilities \cite{5gand6g}. Departing from the traditional Nyquist criterion, which sets limits on the maximum data rate achievable without inter-symbol interference (ISI), FTN challenges conventional wisdom by introducing intentional ISI. By cleverly exploiting this intentional ISI, FTN enables communication rates that are unattainable by Nyquist transmission. 

In classical communication theory, consecutive symbols are transmitted in an orthogonal fashion and do not result in ISI. The Nyquist limit refers to the maximum signaling rate beyond which orthogonality between consecutive symbols can no longer be maintained. In contrast, in \cite{Mazo}, Mazo demonstrated that FTN signals can be transmitted at a rate of $1/\delta$ without an impact on the minimum distance of binary sinc pulses if the acceleration factor $\delta \in [0.802,1]$. In other words, approximately 25\% more bits than Nyquist signaling can be sent within the same bandwidth and at the same signal-to-noise ratio (SNR) without compromising the bit error rate (BER), given ideal processing at the receiver. We define the acceleration rate as $\frac{1}{\delta T}$, which quantifies the relative increase in transmission rate of FTN signaling compared to the conventional Nyquist signaling. It characterizes the ratio between the FTN symbol rate and the Nyquist symbol rate, indicating how much faster FTN transmits symbols.

Numerous studies have explored the capacity of single-input single-output (SISO) FTN systems. In \cite{rusek} and \cite{property}, the authors propose an achievable rate assuming independent transmitted symbols. However, optimal input distributions may introduce correlation to mitigate ISI, as investigated in \cite{bajcsy} and \cite{rusek12}, which propose a waterfilling-based power allocation strategy. These works, however, do not fully address ISI’s impact under transmit power constraints, differing from orthogonal signaling. Transmit precoding offers another approach, as shown in \cite{linearprecoding}, with eigen-decomposition-based precoding schemes proposed in \cite{svd}, \cite{eigendecomposition}, and \cite{chaki}. In contrast, \cite{ganji} provides the complete SISO FTN capacity expression, accounting for ISI in the power constraint and covering all acceleration factors $\delta \in (0,1]$. Similarly, \cite{asympftn} shows that binary FTN signaling is asymptotically optimal as $\delta$ tends to 0.


Using multiple antennas at both the transmitter and receiver significantly enhances channel capacity \cite{telatar}. Since \cite{telatar}, the field has expanded to massive multiple-input multiple-output (MIMO) \cite{lu2014overview} and cell-free massive MIMO \cite{elhoushy2021cell}. Integrating FTN signaling into MIMO systems is a natural next step. Rusek \cite{rusek2009existence} explored the Mazo limit in MIMO, while Modenini et al. \cite{andrea} demonstrated FTN performance gains in large-scale antenna systems with simplified receivers. Yuhas et al. \cite{michael} studied MIMO FTN capacity with independent inputs, but this assumption is limited, as optimal input distribution may not be independent. The ergodic capacity of MIMO FTN over triply-selective Rayleigh fading channels was analyzed in \cite{MIMOFTNfad} without transmitter channel state information. In \cite{zhang2022faster}, we examined FTN capacity in MIMO systems for both frequency-flat and frequency-selective channels, but our analysis did not address small acceleration factors. Specifically, the capacity for $\delta < 1/(1+\beta)$, when using root raised cosine (RRC) pulses with roll-off factor $\beta \in [0,1]$, remains unknown. This work addresses that gap by establishing FTN capacity for MIMO systems at small $\delta$ values, providing new insights into transmit power, {instantaneous-to-average power ratio (IAPR)}, and received SNR in FTN systems.

{In our previous work \cite{zhang2022faster}, we obtained the capacity of MIMO FTN for $\delta\geq\frac{1}{1+\beta}$, where it can be achieved with spatial-domain water-filling and frequency-domain spectrum shaping. In this work, we have extended the analysis to all $\delta$ in $[0,1]$. We compare this paper with \cite{zhang2022faster} and the literature in Table \ref{tab:contribution}.}
\begin{table}[t]
    \centering 
    \renewcommand{\arraystretch}{1.2} 
    \begin{tabular}{|>{}c|>{}c|>{}c|>{}c|}
    \hline
     & This paper & \cite{zhang2022faster} & \cite{ganji} \\
     \hline 
    $\delta\geq\frac{1}{1+\beta}$ & \ding{51} & \ding{51} & \ding{51} \\
    \hline 
    $\delta<\frac{1}{1+\beta}$ & \ding{51} & \ding{55} & \ding{51} \\
    \hline 
    MIMO channel & \ding{51} & \ding{51} & \ding{55} \\
    \hline
\end{tabular}
\vspace{3mm}
\caption{Contribution of this paper compared to other works.}
    \label{tab:contribution}
\end{table}

In FTN signaling, small $\delta$ leads to significant pulse overlap, often resulting in high transmission power peaks. Thus, examining {IAPR} performance in FTN signaling is crucial. Several studies have investigated {IAPR} in FTN systems. For instance, \cite{petitpied2018} compared FTN and Nyquist signaling in terms of spectral efficiency, bandwidth, SNR, and {IAPR}, showing FTN’s optimality over a range of spectral efficiencies. The paper \cite{liu2018peak} demonstrated that multi-carrier FTN signaling could achieve lower {IAPR} than Nyquist systems. However, these studies focused on specific roll-off and acceleration factors and did not provide a universal analysis of FTN {IAPR} performance compared to Nyquist signaling.


FTN has also been proposed as an effective solution to {IAPR} issues in satellite communications. In \cite{lucciardi2016}, FTN was demonstrated to enhance spectral efficiency under non-linear amplification. The paper \cite{delamotte2017faster} showed that FTN could achieve lower {IAPR} than Nyquist transmission when symbol constellations and rates are appropriately chosen for small roll-off factors. Similarly, \cite{liq2020} examined FTN {IAPR} performance in DVB-S2X systems using a high power amplifier model.



While previous studies have made practical assumptions and examined {IAPR} behavior in FTN signaling, this paper aims to analyze {IAPR} across different acceleration factors both for fixed average transmission power and for received SNR. This analysis will offer general design guidelines for practical FTN transmission. Additionally, there is no existing {IAPR} performance analysis for FTN under {extremely small acceleration factors} in the literature. We will demonstrate that there is a trade-off between {IAPR} behavior and received SNR. If received SNR is fixed, this can imply an exploding {IAPR} behavior. Consequently, BER and/or spectral efficiency calculations should consider {IAPR} effects.


	 
The organization of the paper is as follows. For MIMO FTN, we define the system model in Section~\ref{sec:systemmodel}. We develop the mutual information expression in Section~\ref{sec:capacityderivation}. In Section~\ref{sec:paprsec}, we derive the theoretical expression for the {IAPR} distribution for FTN transmission with Gaussian signaling as well as QPSK symbols. In Section~\ref{sec:simresult}, we present our numerical results. Finally, in Section~\ref{sec:conclusion} we conclude the paper.
	
We use the following notations in the paper. The superscript $*$ means complex conjugate
 and $\star$ means convolution.
The superscript $T$ is transpose, the operation $\otimes$ is the Kronecker product. 
The operation $\mathbb{E}$ is expectation and the superscript $\dagger$ is the Hermitian conjugate. The indicator function is denoted as $\mathbbm{1}$. An identity matrix of size $L\times L$ is shown as $\bm{I}_L$, and the trace operation is $\text{tr}(\cdot)$. Finally, $(a)^+$ means $\max(0,a)$.

\vspace{-0.35cm}
\section{System Model}\label{sec:systemmodel}

In this paper, we assume a communication scenario where the transmitter is equipped with $K$ antennas and the receiver is equipped with $L$ antennas. Assuming that $N$ symbols {are} transmitted from each transmit antenna, we use $a_k[n], n=0, \dots, N-1, k=1, \dots, K$, to denote the $n$th transmitted  symbol from the $k$th transmit antenna. The symbols of each antenna go through the pulse shaping filter, which we denote as $p(t)$, and we let all transmit antennas have the same pulse shaping filter. {For FTN transmission, the symbols are transmitted every $\delta T$ seconds, where $T$ is the sampling period in which there will be no ISI at sampling instants.} 
Therefore, we write the expression of the transmitted signal $x_k(t)$ from the $k$th antenna as 
\begin{equation}
    x_k(t)=\sum_{m=0}^{N-1}a_k[m]p(t-m\delta T). \label{eqn:xt}
\end{equation}
After the signal is sent to the wireless channel, each transmission link experiences fading. {We assume in the paper that the communication suffers from frequency-flat fading} and denote the channel coefficient for the link from the $k$th transmit antenna to the $l$th receive antenna, $l=1,\dots, L$, as {$h_{lk}\in\mathbb{C}$}.
 On the receiver side, {the signal at the $l$th receiver antenna including the circularly symmetric complex Gaussian noise, which we denote as $\xi_l(t)$, goes through the matched filter.} By definition, the matched filter is $p^*(-t)$. The output of the matched filter of the $l$th receive antenna $y_l(t)$ is
\begin{equation}
    y_l(t)=\sum_{k=1}^{K}h_{lk}\sum_{m=0}^{N-1}a_k[m]g(t-m\delta T)+\eta_l(t),
\end{equation}
where $g(t)=p(t)\star p^*(-t)$. Moreover, $\eta_l(t)$ can be written as $\eta_l(t)=\xi_l(t)\star p^*(-t)$. After the matched filter, we sample the output $y_l(t)$ at every $\delta T$ seconds.    We write the samples $y_l[n], n=0,\dots, N-1,$ as 
\begin{align}
    y_l[n]&=y_l(n\delta T) \notag \\
    &=\sum_{k=1}^{K}h_{lk}\sum_{m=0}^{N-1}a_k[m]g((n-m))\delta T)+\eta_l(n\delta T) \notag \\ 
    & =\sum_{k=1}^{K}h_{lk}\sum_{m=0}^{N-1}a_k[m]g[n-m]+\eta_l[n].\label{eq:revsamp}
\end{align}

In FTN signaling, we {increase} the symbol rate without changing the pulse shape or the bandwidth. This inevitably leads to ISI since the Nyquist zero-ISI law is violated. This can be shown by the fact that $g((n-m)\delta T)\neq 0$ when $n\neq m$ and thus at sampling instant $n\delta T$, the symbol $a_k[n]$ will receive interference from other symbols due to the non-zero $g((n-m)\delta T)$. We can write \eqref{eq:revsamp} in a vector form as 
\begin{equation}
    \bm{y}_l=\sum_{k=1}^{K}h_{lk}\bm{G}\bm{a}_k+\bm{\eta}_l,
\end{equation}
where $\bm{y}_l=[y_l[0],\dots, y_l[N-1]]^T, \bm{a}_k=[a_k[0],\dots, a_k[N-1]]^T,$ and $\bm{\eta}_l=[\eta_l[0],\dots, \eta_l[N-1]]^T$. The $N \times N$ matrix $\bm{G}$ is formed by $(\bm{G})_{n,m}=g[n-m]$. It is easy to see that the $\bm{G}$ matrix is Hermitian. By collecting the samples from all the receive antennas, we can write the input-output model for the MIMO FTN channel  as 
\begin{equation}
    \bm{Y}=\left(\bm{H}\otimes\bm{G}\right)\bm{A}+\bm{\Omega},
\end{equation}
where $\bm{Y}=[\bm{y}_1^T, \bm{y}_2^T,\dots,\bm{y}_{L}^T]^T, \bm{A}=[\bm{a}_1^T, \bm{a}_2^T,\dots,\bm{a}_{K}^T]^T$, and $\bm{\Omega}=[\bm{\eta}_1^T, \bm{\eta}_2^T,\dots,\bm{\eta}_{L}^T]^T$. The channel matrix $\bm{H}$ contains the channel coefficients for all the transmission links and is defined by $(\bm{H})_{l,k}=h_{lk}$. For ease of notation, we denote the matrix $\bm{H}\otimes \bm{G}$ as $\Tilde{\bm{H}}$. The Gaussian noise vector $\bm{\eta}_l, l=1,\dots, L,$ follows the distribution $\mathcal{CN}\left(\bm{0}_N, \sigma_0^2\bm{G}\right)$, where $\bm{0}_N$ is a zero vector with size $N\times 1$ and $\sigma_0^2$ is the power spectral density (PSD) of $\xi_l(t)$. This shows that due to FTN, the additive noise becomes correlated, meanwhile, in Nyquist signaling, the $\bm{G}$ matrix reduces to the identity matrix and output noise terms become independent. {Since the noise terms $\xi_l(t)$ are independent of each other for all $l$, the matched filter output noise terms $\eta_l(t)$ are also independent of each other for all $l$.} Therefore, the noise vector $\bm{\Omega}$ has the distribution $\mathcal{CN}\left(\bm{0}_{LN}, \sigma_0^2(\bm{I}_L\otimes \bm{G})\right)$. 
\vspace{-0.2cm}

\section{Capacity Derivation} \label{sec:capacityderivation}

{In this section, we first derive the time-domain mutual information expression between the channel input $\bm{A}$ and the channel output $\bm{Y}$. Then we apply the generalized Szeg\"o's theorem to convert the expression into the frequency domain and form the capacity problem in the frequency domain.             Eventually, we find the solution to the optimization problem.}  
\vspace{-0.2cm}

\subsection{Mutual Information Derivation}

In order to find the capacity, we write the  mutual information between the output $\bm{Y}$ and the input $\bm{A}$,  $I(\bm{Y};\bm{A})$ as
\begin{align}
    I(\bm{Y};\bm{A})&=h(\bm{Y})-h(\bm{Y}|\bm{A}) \notag \\
    =&\log_2\det\left(\bm{\Sigma_\Omega}+\Tilde{\bm{H}}\bm{\Sigma_A}\Tilde{\bm{H}}^\dagger\right)-\log_2\det\left(\bm{\Sigma_\Omega}\right). \label{eqn:IYA}
\end{align}
Here, $h(\cdot)$ is the differential entropy. The matrices $\bm{\Sigma_\Omega}=\mathbb{E}\left[\bm{\Omega}\bm{\Omega}^\dagger\right]$, $\bm{\Sigma_A}=\mathbb{E}\left[\bm{A}\bm{A}^\dagger\right]$ are respectively the covariance matrices for the Gaussian noise $\bm{\Omega}$ and the data symbols $\bm{A}$. 
{Since the noise process $\eta_l(t)$ is a stationary zero-mean Gaussian process, the
optimal input is also a stationary zero-mean Gaussian process \cite{Cover}. 
As  $\bm{\Sigma_\Omega}=\sigma_0^2(\bm{I}_L\otimes \bm{G})$,  we calculate  the input covariance matrix $\bm{\Sigma_A}$ as}
\begin{equation}
    \bm{\Sigma_A}=\left[\begin{matrix}
        \bm{\Sigma}_{1,1} & \bm{\Sigma}_{1,2} &\dots & \bm{\Sigma}_{1,K} \\
        \bm{\Sigma}_{2,1} & \bm{\Sigma}_{2,2} &\dots & \bm{\Sigma}_{2,K} \\
        \vdots & \vdots & \ddots & \vdots \\
        \bm{\Sigma}_{K,1} & \bm{\Sigma}_{K,2} &\dots & \bm{\Sigma}_{K,K}
    \end{matrix}\right],
\end{equation}
where $\bm{\Sigma}_{i,j}=\mathbb{E}\left[\bm{a}_i\bm{a}_j^\dagger\right], i,j=1,\dots,K$ and {we can see that $\bm{\Sigma_A}$ is a block matrix}. Since the input processes to the transmit antennas are zero-mean Gaussian processes, the input to each antenna is itself stationary, and inputs of any transmit antenna pairs are jointly stationary. In other words, the autocorrelation function  $R_{k,k}[n,m]=\mathbb{E}[a_k[n]a^*_k[m]]=R_{k,k}[n-m]$ and the cross-correlation function $R_{k,l}[n,m]=\mathbb{E}[a_k[n]a^*_l[m]]=R_{k,l}[n-m]$. Therefore, each block inside the block matrix $\bm{\Sigma_A}$ is a Toeplitz matrix. The entries of an $N\times N$ Toeplitz matrix $\bm{R}$ has the property that $\left(\bm{R}\right)_{i,j}=r_{i-j}, i,j=0,\dots, {N-1}$, in other words, a Toeplitz matrix has the same value on each diagonal. Furthermore, a block Toeplitz matrix is a block matrix where each of its blocks is a Toeplitz matrix. Another important concept we will be using is the generating function of a Toeplitz matrix. It is defined as
\begin{equation}
    \mathcal{G}(\bm{R})=\sum_{k=-\infty}^{\infty}r_ke^{j2\pi f_nk}, f_n\in \left[-\frac{1}{2},\frac{1}{2}\right].\label{eqn:defgenfunc}
\end{equation}

The capacity of the MIMO FTN channel can be obtained by finding the maximum of the asymptotic average mutual information over all input distributions $p(\bm{A})$, namely, 
\begin{equation}
    C_{FTN}=\underset{p(\bm{A})}{\max}\underset{N\rightarrow\infty}{\lim}\frac{1}{N}I(\bm{Y};\bm{A}).\label{eqn:capfirst}
\end{equation}
 The capacity can be found by first taking the limit operation and then finding the optimal input distribution.
 In order to calculate \eqref{eqn:capfirst}, we need to first manipulate the expression \eqref{eqn:IYA}. 
According to \cite{zhang2022faster}, \eqref{eqn:IYA} can be written as
\begin{equation}
    I(\bm{Y},\bm{A})=\log_2\det\left(\bm{I}_{KN}+\sigma_0^{-2}\bm{\Sigma_A}\left(\bm{H}^\dagger\bm{H}\otimes\bm{G}\right)\right), \label{eqn:mutinfsimp}
\end{equation}
where the detailed derivation can be found in \cite[(21)-(29)]{zhang2022faster}. It is straightforward to see that the matrix $\left(\bm{H}^\dagger\bm{H}\otimes\bm{G}\right)$ is a block Toeplitz matrix since $\bm{G}$ itself is Toeplitz. We know from \cite[Theorem 2]{jesus} that the product of block Toeplitz matrices is asymptotically Toeplitz, so the matrix product $\bm{\Sigma_A}\left(\bm{H}^\dagger\bm{H}\otimes\bm{G}\right)$ inside the determinant in \eqref{eqn:mutinfsimp} is asymptotically Toeplitz. Similarly, the matrix $\bm{I}_{LN}+\sigma_0^{-2}\bm{\Sigma_A}\left(\bm{H}^\dagger\bm{H}\otimes\bm{G}\right)$ is also asymptotically Toeplitz. 
%
Using this fact, we can find the limit of $\underset{N\rightarrow\infty}{\lim}\frac{1}{N}I(\bm{Y};\bm{A})$ 
 by invoking the generalized Szeg\"o's theorem \cite[Theorem 3]{ganji2018novel}, which is  given in Lemma \ref{lem:gensze}.
\begin{lemma}{\cite[Theorem 3]{ganji2018novel}}
\label{lem:gensze}
    Assume $\bm{T}$ is a $NK\times NK$ block Toeplitz matrix with the structure 
    \begin{equation}
        \bm{T}=\left[\begin{matrix}
        \bm{T}_{1,1} & \bm{T}_{1,2} &\dots & \bm{T}_{1,K} \\
        \bm{T}_{2,1} & \bm{T}_{2,2} &\dots & \bm{T}_{2,K} \\
        \vdots & \vdots & \ddots & \vdots \\
        \bm{T}_{K,1} & \bm{T}_{K,2} &\dots & \bm{T}_{K,K}
    \end{matrix}\right],
    \end{equation}
    where $\bm{T}_{i,j}$ are $N\times N$ Toeplitz matrices. We have 
    \begin{equation}
        \underset{N\rightarrow\infty}{\lim} \frac{1}{N}\sum_{i=1}^{KN}F\left(\lambda_i(\bm{T})\right)=\int_{-\frac{1}{2}}^{\frac{1}{2}}\sum_{k=1}^KF\left(\lambda_j(\bm{T}(f_n))\right)df_n,
    \end{equation}
    where $\lambda_i(\cdot)$ means the $i$th eigenvalue of the matrix inside the parenthesis, and $F(\cdot)$ is a  continuous function defined over the range of $f_n$. With a slight abuse of notation, which can be confused with the matrix $\bm{T}$ itself, We denote the generating matrix of the block Toeplitz matrix $\bm{T}$ as $\bm{T}(f_n)$.
    The $K\times K$ matrix $\bm{T}(f_n)$ is composed of the generating functions of the Toeplitz matrices $\bm{T}_{i,j}, i,j=1,\dots, K$. The entries of $\bm{T}(f_n)$ are calculated as $\left(\bm{T}(f_n)\right)_{i,j}=\mathcal{G}(\bm{T}_{i,j})$. 
\end{lemma}

\begin{lemma}{\cite[Theorem 2]{jesus}}\label{lem:toepmprod}
    The generating matrix of the product of two block Toeplitz matrices is the product of the generating matrices of these two matrices.
\end{lemma}
\begin{lemma}\label{lem:toepmsum}
    The generating matrix of the sum of two block Toeplitz matrices is the sum of the generating matrices of these two matrices.
\end{lemma}

 By applying Lemma \ref{lem:gensze}, we have the derivation below 
 \begin{align}
     \underset{N\rightarrow\infty}{\lim}\frac{1}{N}I(\bm{Y};\bm{A})=&\int_{-\frac{1}{2}}^{\frac{1}{2}}\log_2\det\big(\bm{I}_{KN}(f_n)\notag\\
     &+\sigma_0^{-2}\bm{\Sigma_A}(f_n)\bm{H}^\dagger\bm{H}G_d(f_n)\big)df_n, \label{eqn:limend}
 \end{align}  due to Lemma \ref{lem:toepmprod} and \ref{lem:toepmsum}, and the fact that the matrices $\bm{I}_{KN}(f_n)$, $\bm{\Sigma_A}(f_n)$,  and $\bm{H}^\dagger\bm{H}G_d(f_n)$ are respectively the generating matrices for $\bm{I}_{KN}$, $\bm{\Sigma_A}(f_n)$, and $\bm{H}^\dagger\bm{H}G_d(f_n)$. 
From \eqref{eqn:defgenfunc} we know that the generating matrix of $\bm{I}_{KN}$ is $\bm{I}_K$, and the generating function of $\bm{G}$ is $G_d(f_n)$, the definition of which can be found in \cite[(90)]{zhang2022faster}. The function  $G_d(f_n)$ is also called the folded spectrum in the literature of FTN \cite{rusek}.
 The folded spectrum is periodical with period $1$, and the shape of it changes with the acceleration factor $\delta$. When an RRC pulse with roll-off factor $\beta$ is used for $p(t)$, $G(f)$, which is the continuous time Fourier transform of $g(t)$, will be duplicated and then shifted by $\frac{m}{\delta T}, m\in\mathbb{Z}$. In the end, the spectrum will be scaled by $\frac {1}{\delta T}$ both in frequency and amplitude. When $p(t)$ is an RRC pulse, the spectrum $G(f)$ is non-zero from $-\frac{(1+\beta)}{2T}$ to $\frac{(1+\beta)}{2T}$. Therefore as $\delta$ becomes smaller than $\frac{1}{1+\beta}$, the support of the folded-spectrum $G_d(f_n)$ will be from $-\frac{(1+\beta)}{2T}$ to $\frac{(1+\beta)}{2T}$ and there will be zero parts in $[-\frac{1}{2},\frac{1}{2}]$. 
\begin{figure}
    \centering
    \includegraphics[width=0.7\linewidth]{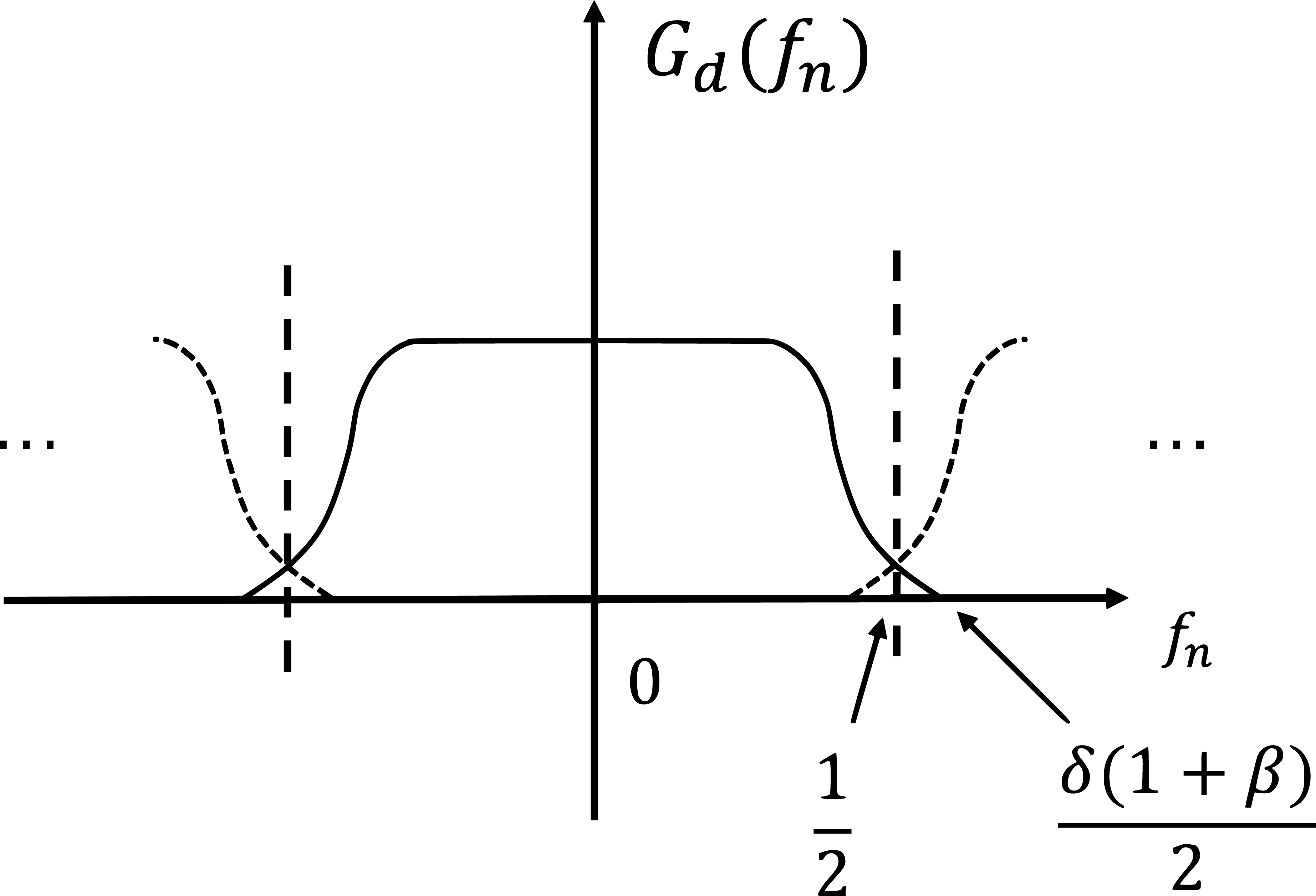}
    \caption{Folded spectrum for RRC pulse with roll-off factor $\beta$.}
    \label{fig:foldspec}
\end{figure}
In Fig. \ref{fig:foldspec}, we give an example of $G_d(f_n)$ where an RRC pulse is used with roll-off factor $\beta$. For ease of exposition, we assume that the pulse $p(t)$ in this paper is also an RRC function with roll-off factor $\beta$.  
{However, the discussion in this paper is not confined to RRC pulses, the shaping pulse $p(t)$ can be switched to other pulse shapes according to the practical design constraints. Furthermore, the analysis in this paper can also be extended to nonorthogonal waveforms \cite{wunan1, wunan2}.}

We next define $\bm{W}=\bm{H}^\dagger\bm{H}$. Since $\bm{W}$ is a Hermitian matrix, it has the eigenvalue decomposition $\bm{W}=\bm{U}\bm{\Gamma}\bm{U}^\dagger$, where $\bm{U}$ is a unitary matrix and $\bm{\Gamma}$ is a diagonal matrix.  The eigenvalues of $\bm{W}$, $\tau_i, i=1,\dots, K$, are on the main diagonal, i.e.,  $\bm{\Gamma}=\text{diag}\left[\tau_1, \tau_2,\dots, \tau_{K}\right]$. 
Then we can upper bound  \eqref{eqn:limend}  by applying the generalized Hadamard inequality as in \eqref{eqn:ub}. 
\begin{eqnarray}
    \lefteqn{\underset{N\rightarrow\infty}{\lim}\frac{1}{N}I(\bm{Y};\bm{A})} \nonumber \\
    &=&\int_{-\frac{1}{2}}^{\frac{1}{2}}\log_2\det\left(\bm{I}_K+\sigma_0^{-2}G_d(f_n)\bm{U}^\dagger\bm{\Sigma_A}(f_n)\bm{U}\bm{\Gamma}\right)df_n  \notag\\
    &\leq& \int_{-\frac{1}{2}}^{\frac{1}{2}}\log_2\det\left(\bm{I}_K+\sigma_0^{-2}G_d(f_n)\bm{\Phi}(f_n)\bm{\Gamma}\right)df_n. \label{eqn:ub}
\end{eqnarray}
The upper bound will be achieved if $\bm{\Sigma_A}(f_n)$ can be diagonalized into $\bm{\Sigma_A}(f_n)=\bm{U}\bm{\Phi}(f_n)\bm{U}^\dagger$, where $\bm{\Phi}=\text{diag}\left[\phi_1(f_n), \phi_2(f_n), \dots,\phi_{K}(f_n)\right]$. We call the functions $\phi_k(f_n), k=1,\dots, K$, as data spectrum, since they are only related to the distribution of the input data symbols. Therefore, \eqref{eqn:ub} can be written as 
\begin{align}
 \lefteqn{   \underset{N\rightarrow\infty}{\lim}\frac{1}{N}I(\bm{Y};\bm{A})} \nonumber \\
 &=&\sum_{k=1}^K\int_{-\frac{1}{2}}^{\frac{1}{2}}\log_2\left(1+\sigma_0^{-2}G_d(f_n)\phi_k(f_n)\tau_k\right)df_n. \label{eqn:obj}
\end{align}
We then proceed to obtain the power constraint. Assume that the transmission has the power limit of $P$. The power constraint expression for transmitting $N$ symbols is 
\begin{align}
    P_{TX}&=\mathbb{E}\left[\frac{1}{N\delta T}\sum_{k=1}^{K}\int_{-\infty}^{\infty}|x_k(t)|^2dt\right] \notag\\
    &=\frac{1}{N\delta T}\text{tr}\left(\left(\bm{I}\otimes\bm{G}\right)\bm{\Sigma_A}\right)\leq P. \label{eqn:powconstres}
\end{align}
 The detailed derivation of \eqref{eqn:powconstres} can be found in (33)-(37) of \cite{zhang2022faster}. 
As the number of symbols goes to infinity, we can apply Szeg\"o's theorem again,
\begin{eqnarray}
    \lefteqn{\underset{N\rightarrow\infty}{\lim}\frac{1}{N\delta T}\text{tr}\left(\left(\bm{I}\otimes\bm{G}\right)\bm{\Sigma_A}\right)} \nonumber \\
    &=&\frac{1}{\delta T}\int_{-\frac{1}{2}}^{\frac{1}{2}}\sum_{k=1}^KG_d(f_n)\phi_k(f_n)df_n.\label{eqn:powcons}
\end{eqnarray}
Eventually, we are able to form our optimization problem for the channel capacity by combining \eqref{eqn:obj} and \eqref{eqn:powcons}, and write
    \begin{align}
        \lefteqn{C_{FTN}(P,\delta)}\\       &=\underset{\substack{\phi_k(f_n),\\k=1,\dots,K}}{\max}\sum_{k=1}^K\int_{-\frac{1}{2}}^{\frac{1}{2}}\log_2\left(1+\sigma_0^{-2}G_d(f_n)\phi_k(f_n)\tau_k\right)df_n \notag\\
        &s.t.~\frac{1}{\delta T}\int_{-\frac{1}{2}}^{\frac{1}{2}}\sum_{k=1}^KG_d(f_n)\phi_k(f_n)df_n\leq P. \notag
    \end{align}
Note that as $\delta<\frac{1}{1+\beta}$, there will be zero parts for the folded spectrum in $\left[-\frac{1}{2},\frac{1}{2}\right]$. Therefore we should perform the optimization on the support of the folded spectrum. We denote the support of $G_d(f_n)$ in $\left[-\frac{1}{2},\frac{1}{2}\right]$ as $\mathcal{S}$ and the optimization problem becomes
    \begin{align}
        \lefteqn{C_{FTN}(P,\delta)}\label{eqn:capobj}\\
        &=\underset{\substack{\phi_k(f_n),\\k=1,\dots,K}}{\max}\sum_{k=1}^K\int_{\mathcal{S}}\log_2\left(1+\sigma_0^{-2}G_d(f_n)\phi_k(f_n)\tau_k\right)df_n \notag \\
        &s.t.~\frac{1}{\delta T}\int_{\mathcal{S}}\sum_{k=1}^KG_d(f_n)\phi_k(f_n)df_n\leq P.\notag
    \end{align}
We write the Karush–Kuhn–Tucker (KKT) conditions as 
    \begin{align}
        \frac{\sigma_0^{-2}G_d(f_n)\tau_k}{1+\sigma_0^{-2}G_d(f_n)\tau_k}-\mu G_d(f_n)-v_k(f_n)&=0 \\
        \mu\left(\frac{1}{\delta T}\int_{\mathcal{S}}\sum_{k=1}^KG_d(f_n)\phi_k(f_n)df_n- P\right)&=0 \\
        v_k(f_n)\phi_k(f_n)&=0,
    \end{align}
where $\mu$ and $v_k(f_n)$'s are the Lagrange multipliers. The solution to this problem is 
\begin{equation}
   \bar{\phi}_k(f_n)=\begin{cases}
          \frac{\delta T}{G_d(f_n)}\left(\frac{1}{\mu}-\frac{1}{\tau_k}\right)^+,& \frac{1}{1+\beta}\leq\delta\leq1  \\
          \frac{ T}{G_d(f_n)(1+\beta)}\left(\frac{1}{\mu}-\frac{1}{\tau_k}\right)^+,& 0<\delta<\frac{1}{1+\beta}
    \end{cases}, \label{eqn:optsolu}
\end{equation}
for $k=1,\dots,K, f_n\in \mathcal{S}$, where $\bar{\phi}_k(f_n)$ means the optimum $\phi_k(f_n)$. The Lagrange multiplier $\mu$ can be found by solving
\begin{equation}
    \frac{1}{\delta T}\int_{\mathcal{S}}\sum_{k=1}^K\left(\frac{1}{\mu}-\frac{1}{\tau_k}\right)^+df_n=P.\label{eqn:mu}
\end{equation}
In order to find the channel capacity, we need to plug the optimum solution \eqref{eqn:optsolu} back into \eqref{eqn:capobj}. For ease of representation,  we call the power allocated to the $k$th eigen-channel $\left(\frac{1}{\mu}-\frac{1}{\tau_k}\right)^+$ as $\sigma^2_k$, {and $\sum_{k=1}^K\sigma^2_k=P$}. We denote the Lebesgue measure of $\mathcal{S}$ as $|\mathcal{S}|$. When $\delta\geq\frac{1}{1+\beta}$, $|\mathcal{S}|=1$, then combining \eqref{eqn:capobj} and \eqref{eqn:optsolu}, the capacity expression is written as  
\begin{equation}
    C_{FTN}(P,\delta)=\sum_{k=1}^K\log_2\left(1+\frac{\sigma^2_k\delta T\tau_k}{\sigma_0^2}\right), \text{for}~\frac{1}{1+\beta}\leq\delta\leq1. \label{eqn:capbdbps}
\end{equation}
   For RRC pulses with roll-off factors that satisfy $\delta<\frac{1}{1+\beta}$, the Lebesgue measure  $|\mathcal{S}|=\delta(1+\beta)$ and $|\mathcal{S}|<1$. Then, the capacity expression becomes
\begin{equation}
    C_{FTN}(P,\delta)=\delta(1+\beta)\sum_{k=1}^K\log_2\left(1+\frac{\sigma^2_k T\tau_k}{\sigma_0^2(1+\beta)}\right),  \label{eqn:capsdbps}
\end{equation}
for $0<\delta<\frac{1}{1+\beta}$. Both the capacity in \eqref{eqn:capbdbps} and \eqref{eqn:capsdbps} are in bits/symbol. We normalize the capacity to convert the unit to bits/s/Hz and induce the following theorem.
\begin{theorem}
\label{thm:them1}
    The capacity (in bits/s/Hz) of the MIMO FTN channel with $K$ transmit and $L$ receive antennas, employing RRC pulses with roll-off factor $\beta$ is equal to 
    \begin{align}
    \lefteqn{C_{FTN}(P,\delta)}\notag\\
    &=\begin{cases}
        \frac{1}{\delta(1+\beta)}\sum_{k=1}^K\log_2\left(1+\frac{\sigma^2_k\delta T\tau_k}{\sigma_0^2}\right),& ~\frac{1}{1+\beta}\leq\delta\leq1  \\
        \sum_{k=1}^K\log_2\left(1+\frac{\sigma^2_k T\tau_k}{\sigma_0^2(1+\beta)}\right),& ~0<\delta<\frac{1}{1+\beta}
    \end{cases}.\vspace{-5mm}
    \end{align}
\end{theorem}
\begin{remark}
    The power constraint in \eqref{eqn:powconstres} means fixed average transmission power. For this constraint, when $0<\delta<\frac{1}{1+\beta}$, the capacity is independent of $\delta$. The capacity for $\frac{1}{1+\beta}\leq\delta\leq1$ is identical with the result in \cite{zhang2022faster}. Therefore, the capacity increases as $\delta$ decreases, and when $\delta$ is at the threshold value $\frac{1}{1+\beta}$, the capacity reaches its maximum and stays the same as $\delta$ keeps decreasing to 0. {As $\delta$ decreases below the threshold, the amount of information each symbol can carry also decreases due to severe ISI. However, since the signaling rate increases at the same time, it compensates for the decreased information per symbol. In other words, although the capacity decreases in bits/symbol, after normalization by signaling rate, the resulting capacity in bits/s/Hz stays the same.} 
\end{remark}
\begin{remark}
    For FTN, the capacity is 0 for $\delta=0$.  If $\delta=0$, all the symbols will be sent all at once and it is impossible to sample each symbol separately. This case then becomes equivalent to sending only one symbol and the capacity of sending a single symbol is known to be zero.  
\end{remark}
\begin{remark}
    The capacity-achieving input distribution is obtained by the combination of 
 spatial domain water-filling and frequency domain spectrum inversion.  The
    power allocated to the $k$th eigen-channel in the optimum data spectrum in \eqref{eqn:optsolu} is determined by spatial domain water-filling, and the shape of the spectrum for each eigenchannel is determined by the frequency domain inverted spectrum $\frac{1}{G_d(f_n)}$.
\end{remark}

\subsection{Different Power Allocation Schemes}\label{sec:diffpowallo}
We have obtained the optimal power allocation over the spatial and frequency domains in \eqref{eqn:optsolu}, which means we perform water-filling in the spatial domain and channel inversion in the frequency domain. We denote this scheme as $O_sO_f$, meaning optimal power allocation both in the spatial and the frequency domains. {In \cite{zhang2022faster}, MIMO FTN capacity for $\delta\geq\frac{1}{1+\beta}$ is derived, and the capacity-achieving power allocation scheme is the same as $O_sO_f$ in \eqref{eqn:optsolu} for $\delta\geq\frac{1}{1+\beta}$.}  We now proceed to investigate other power allocation schemes.  
\subsubsection{Suboptimal in space and optimal in frequency}\label{sec:ssof}
Secondly, we study uniform power allocation in the spatial domain and optimal power allocation in the frequency domain. {We denote this scheme as $S_sO_f$ as suboptimal in the spatial domain and optimal in the frequency domain.} In this case, for $f_n\in \mathcal{S}$, the data spectrum $\phi_k(f_n), k=1,\dots, K$, follows 
\begin{equation}
\phi_k(f_n)=\left\{
\begin{matrix}
     \frac{P}{K}\frac{\delta T}{G_d(f_n)}, &\frac{1}{1+\beta}\leq\delta\leq1  \\
      \frac{P}{K}\frac{ T}{G_d(f_n)(1+\beta)}, &0<\delta<\frac{1}{1+\beta}
\end{matrix}\right.  . \label{eqn:ssof}
\end{equation}
\subsubsection{Optimal in space and suboptimal in frequency}\label{sec:ossf}
In the third scheme, the data spectrum follows
\begin{equation}
    \phi_k(f_n)=\left\{
\begin{matrix}
     P_k\delta T,  &\frac{1}{1+\beta}\leq\delta\leq1  \\
      \frac{P_kT}{1+\beta}, &0<\delta<\frac{1}{1+\beta}
\end{matrix}\right.  , \label{eqn:ossf}
\end{equation}
for $f_n\in \mathcal{S}$. {We denote this scheme as $O_sS_f$ meaning optimal in the spatial domain and suboptimal in the frequency domain.}
\subsubsection{Suboptimal in space and suboptimal in frequency}\label{sec:uniformpowallo}
Finally, in the fourth scheme, we perform uniform power allocation both in spatial and frequency domains. It is denoted as $S_sS_f$ and for $f_n\in \mathcal{S}$, its data spectrum is given as
\begin{equation}
    \phi_k(f_n)=\left\{
\begin{matrix}
     \frac{P\delta T}{K},  &\frac{1}{1+\beta}\leq\delta\leq1  \\
      \frac{PT}{K(1+\beta)}, &0<\delta<\frac{1}{1+\beta}
\end{matrix}\right.  . \label{eqn:sssf}
\end{equation}
We will be comparing all these three schemes with the optimal scheme in Section~\ref{sec:simresult}, where we present the numerical results.

{
\subsection{MIMO FTN Capacity under Frequency Selective (FS) Fading Channels}

In real-life wireless communication scenarios, the signal often faces scattering environments such as buildings and plants, the receiver will receive multiple reflections of the same signal. In this case, the channel becomes multipath channel, and its effect in the frequency domain is frequency-selectivity. Therefore, in this section, we investigate the capacity of MIMO FTN under frequency-selective fading for $\delta<\frac{1}{1+\beta}$. 

We assume that the transmitted signal arrives at the receiver from different paths, each with a different delay and gain. Then the channel can be modeled by a tapped-delay filter with $J$ taps. We assume the signals from all the transmit antennas experience the same delay, and the $j$th tap has delay $d_j$, then the channel link from the $k$th transmit antenna to the $l$th receive antenna $h_{lk}(t)$ has expression 
\begin{equation}
    h_{lk}(t)=\sum_{j=0}^{J-1}h^j_{lk}\Delta(t-d_j),
\end{equation}
where we represent $\Delta(\cdot)$ as the Dirac delta function and $h^j_{lk}$ is the channel gain from the $l$th transmit antenna to the $k$th receive antenna for the $j$th tap. 
Note that in \cite{zhang2022faster}, the channel taps are assumed to be integer multiples of the symbol period $\delta T$. In this paper, we do not have limitations on the tap delay $d_j$, which is a more general assumption. If we let $d_j=j\delta T, j=0,\dots, J-1$, our system model boils down to the system model in \cite{zhang2022faster}. 
The output of the matched filter at the $l$th receive antenna induced by the $k$th transmit antenna has the expression 
	 	\begin{equation}
		y_{lk}(t)=\sum_{n=0}^{N-1}\sum_{j=0}^{J-1}h^j_{lk}a_k[n]g(t-n\delta T-d_j). \label{eqn:fsfctimemodel}
	\end{equation}
	 Then we sample the output at every $\delta T$ seconds,  the resulting $N$ samples can be written in vector form as 
	\begin{align}
		\bm{y}_{lk}=
			\left[\begin{matrix}
				y_{lk}[0] \\	y_{lk}[1]\\ \vdots \\ 	y_{lk}[N-1] 
			\end{matrix}\right] = \sum_{j=0}^{J-1}h^j_{lk}\bm{G}^j\bm{a}_k. \label{eqn:tildeGkl}
	\end{align}
	We define $\bm{G}^j$, which is a $j$-\emph{shifted} version of $\bm{G}$, as
	$(\bm{G}^j)_{m,n}=g((m-n)\delta T-d_j)$
	with $\bm{G}^0=\bm{G}$. Thus, the samples at the output of the matched filters at all receive antennas become
	\begin{align}
		\bm{Y}
	 &= \left[\sum_{j=0}^{J-1}(\tilde{\bm{H}}^j\otimes\bm{G}^j)\right]\bm{A} +\bm{\Omega}, \label{eqn:71} 
	\end{align}
	where $(\tilde{\bm{H}}^j)_{l,k}=h^j_{lk}$.  
	We derive the capacity expression in FS fading channels as 
	in \eqref{eqn:cfsmanipulated}. 
    \begin{figure*}
        \begin{equation}
    C^{FS}=\underset{N\rightarrow\infty}{\lim}\underset{\bm{\Sigma_A,\mathcal{P}}}{\max}\frac{1}{N}\log_2\det\bigg[\bm{I}_{LN}+\sigma^{-2}_0\bm{\Sigma_A}\sum_{j=0}^{J-1}\sum_{i=0}^{J-1}\left(\tilde{\bm{H}}^{j\dagger}\tilde{\bm{H}}^{i}\otimes\bm{G}^{j\dagger}\bm{G}^{-1}\bm{G}^{i}\right)\bigg].  \label{eqn:cfsmanipulated}
\end{equation}
    \end{figure*}
We call the matrix $\sum_{j=0}^{J-1}\sum_{i=0}^{J-1}\left(\tilde{\bm{H}}^{j\dagger}\tilde{\bm{H}}^{i}\otimes\bm{G}^{j\dagger}\bm{G}^{-1}\bm{G}^{i}\right)$ as $\bm{G_H}$. 
Note that $\bm{G_H}$ is an asymptotically block Toeplitz matrix. In order to find its generating matrix, we need to find the generating function of  $\bm{G}^{j\dagger}\bm{G}^{-1}\bm{G}^{i}$.
The generating function of $\mathcal{G}(\bm{G}^{j\dagger}\bm{G}^{-1}\bm{G}^{i})$ is calculated as
    \begin{align}
        &\mathcal{G}(\bm{G}^{j\dagger}\bm{G}^{-1}\bm{G}^{i})=
    \left(\sum_ng\left(-n\delta T-d_j\right)e^{j2\pi f_nn}\right)\times \notag\\
    &\quad\quad\quad\quad\left(\sum_mg\left(m\delta T-d_i\right)e^{j2\pi f_nm}\right)G_d^{-1}(f_n),
    \end{align}
    for $f_n\in[-\frac{1}{2},\frac{1}{2}]$. Then 
    \begin{align}
        \mathcal{G}(\bm{G}^{j\dagger}\bm{G}^{-1}\bm{G}^{i})&=\left(\frac{1}{\delta T}G\left(\frac{f_n}{\delta T}\right)e^{-j2\pi \frac{f_n}{\delta T}d_j}\right)\times\notag\\
        &\left(\frac{1}{\delta T}G\left(\frac{f_n}{\delta T}\right)e^{j2\pi \frac{f_n}{\delta T}d_i}\right)\left(\frac{1}{\delta T}G\left(\frac{f_n}{\delta T}\right) \right)^{-1} \notag\\
    &=\frac{1}{\delta T}G\left(\frac{f_n}{\delta T}\right)e^{j2\pi \frac{f_n}{\delta T}(d_i-d_j)}.
    \end{align}
    
We let $\mathcal{G}$ represent both finding the generating function and finding the generating matrix operations. 
	 With the derivation above we can easily obtain the generating matrix of $\bm{G_H}$, whose $(n,m)$th entry is
\begin{align}
    \lefteqn{\big(\mathcal{G}(\bm{G_H})\big)_{n,m}
    =\sum_{i,j}\left(\sum_{k=1}^{K}h^{j*}_{nk}h^{i}_{mk}\right)\mathcal{G}(\bm{G}^{j\dagger}\bm{G}^{-1}\bm{G}^{i})} \notag\\
    &=\sum_{k=1}^{K}\left[\left(\sum_{i=0}^{J-1}h^{i}_{mk}e^{j2\pi \frac{f_n}{\delta T}d_i}\right)\left(\sum_{j=0}^{J-1}h^{j*}_{nk}e^{-j2\pi \frac{f_n}{\delta T}d_j}\right)\frac{1}{\delta T}G\left(\frac{f_n}{\delta T}\right)\right] \notag \\
    &=\frac{1}{\delta T}\sum_{k=1}^{K}H^*_{nk}(-f_n)H_{mk}(-f_n)G\left(\frac{f_n}{\delta T}\right),
\end{align}
where $H_{km}(-f_n)$ is defined as $H_{km}(-f_n)=\sum_{i=0}^{J-1}h^{i}_{km}e^{-j2\pi \frac{f_n}{\delta T}d_i}$. This is the continuous time Fourier transform of the channel $h_{lk}(t)$. We then perform frequency scaling $f=f_n/\delta T$. We call this as the link spectrum. We now collect all the link spectrum and form the channel spectrum matrix $\tilde{\bm{H}}(-f_n)$, namely, $\left(\tilde{\bm{H}}(-f_n)\right)_{l,k}=H_{lk}(-f_n)$.  Then it is straightforward to write 
\begin{equation}
    \mathcal{G}(\bm{G_H})=\frac{1}{\delta T}G\left(\frac{f_n}{\delta T}\right)\tilde{\bm{H}}(-f_n)^\dagger\tilde{\bm{H}}(-f_n)=\frac{1}{\delta T}G\left(\frac{f_n}{\delta T}\right)\tilde{\bm{Z}}(-f_n),
\end{equation}
where $\tilde{\bm{Z}}(-f_n)\triangleq\tilde{\bm{H}}(-f_n)^\dagger\tilde{\bm{H}}(-f_n)$.

 By combining the objective \eqref{eqn:cfsmanipulated} and the power constraint \eqref{eqn:powcons}, the capacity optimization problem becomes 
\begin{align}
    &C^{FS}=\underset{\bm{\Sigma_A}(f_n)}{\max}\int_{\mathcal{S}}\log_2\det\bigg[\bm{I}_{L}+\frac{G\left(\frac{f_n}{\delta T}\right)}{\sigma^{2}_0\delta T}\bm{\Sigma_A}(f_n)\tilde{\bm{Z}}(f_n)\bigg]df_n \label{eqn:fsfcobj}\\
    & s.t. \quad\quad \frac{1}{(\delta T)^2}\int_{\mathcal{S}}\text{tr}\big[G\left(\frac{f_n}{\delta T}\right)\bm{\Sigma_A}(f_n)\big]df_n\leq P.\label{eqn:fsfccons}
\end{align}
As done in Section \ref{sec:capacityderivation}, we again only perform power allocation over the support. Let's define  $\tilde{\bm{\Sigma}}_{\bm{A}}(f_n)  \triangleq \frac{1}{\delta T}G\left(\frac{f_n}{\delta T}\right)\bm{\Sigma_A}(f_n)$.    
Next, we use eigenvalue decomposition to diagonalize $\tilde{\bm{\Sigma}}_{\bm{A}}(f_n)$ and $\tilde{\bm{Z}}(f_n)$,  
\begin{align}
    \tilde{\bm{\Sigma}}_{\bm{A}}(f_n)&=\tilde{\bm{U}}(f_n)\tilde{\bm{\Phi}}(f_n)\tilde{\bm{U}}(f_n)^\dagger \label{eqn:decompS}\\
    \tilde{\bm{Z}}(f_n)&=\tilde{\bm{V}}(f_n)\tilde{\bm{T}}(f_n)\tilde{\bm{V}}(f_n)^\dagger,\label{eqn:decompW}
\end{align}
where the diagonal matrices $\tilde{\bm{\Phi}}(f_n)$ and $\tilde{\bm{T}}(f_n)$ have the structure $\tilde{\bm{\Phi}}(f_n)=\text{diag}\{\phi_1(f_n), \phi_2(f_n),\dots, \phi_K(f_n)\}$ and $\tilde{\bm{T}}(f_n)=\text{diag}\{\tau_1(f_n), \tau_2(f_n),\dots, \tau_K(f_n)\}$. 
The channel matrix $\tilde{\bm{Z}}(f_n)$ is decomposed into its eigenchannels, where $\tau_i(f_n)$ denotes the eigenmodes at frequency $f_n$, and $\phi_i(f_n)$ represents the power allocated to each eigenchannel, which is referred to as the eigenspectrum. 
The upper bound of the objective \eqref{eqn:fsfcobj} is achieved when 
 $\tilde{\bm{U}}(f_n)=\tilde{\bm{V}}(f_n),  \forall f_n$. Then  the optimal solution for $\phi_i(f_n)$ is
\begin{equation}
    \phi_i(f_n)^*=\left(\frac{1}{\mu}-\frac{1}{\tau_i(f_n)}\right)^+, i=1,\dots,L, \quad f_n\in\mathcal{S}, \label{eqn:solufsfc}
\end{equation}
where we use the superscript $*$ to represent optimal. We find $\mu$ by solving 
\begin{equation}
    \frac{1}{\delta T}\int_{\mathcal{S}}\sum_{i=1}^{L}\left(\frac{1}{\mu}-\frac{1}{\tau_i(f_n)}\right)^+df_n = P. \label{eqn:fsfcoptpowallo}
\end{equation} 
The intuition of the power allocation scheme is that water-filling is performed in each eigenchannel $\tau_k(f_n)$ where the water level is the same for all the eigenchannels.  However, for $\delta<\frac{1}{1+\beta}$, the power allocation is performed over the support.
The optimal generating matrix $\bm{\Sigma_A}(f_n)$ can be obtained 
by $\bm{\Sigma_A}(f_n)=\frac{\delta T\tilde{\bm{\Sigma}}_{\bm{A}}(f_n)}{G\left(\frac{f_n}{\delta T}\right)}, f_n\in\mathcal{S}. $ 
}

\section{Peak to Average Power Ratio Analysis}
\label{sec:paprsec}

Due to {small acceleration factor}, the pulses of FTN transmission are tightly packed, causing a high chance that pulses overlap with each other resulting in high {IAPR}. In practice, the power amplifier at the transmitter has a certain threshold output value called the saturation point, which limits the maximum amplitude of the output signal. 
In order to maintain linear performance at the power amplifier,
the input power needs to be reduced to accommodate signal peaks. The reduction amount is referred to as the back-off value and is measured in decibels (dB). Without any compensation techniques, the required power amplifier back-off approximately equals the {IAPR} of the input signal.

In FTN signaling, with transmit power $P$,  transmitting $N$ symbols takes  $N\delta T$ seconds and the overall energy is equal to $NP\delta T$. Also, we find that the average energy per symbol is $E=P\delta T$. In FTN signaling, the energy for each symbol decreases as $\delta$ decreases for fixed transmission power $P$. For practical constellations, this implies smaller minimum Euclidean distances and results in higher error probabilities. 

In FTN, we need two definitions for SNR. We define transmit SNR as $\mathsf{SNR_{tx}}=\frac{P}{\sigma_0^2}$, which is the transmit power over noise variance. We also define received SNR as $\mathsf{SNR_{rx}}=\frac{E/T}{\sigma_0^2}=\frac{P\delta }{\sigma_0^2}$. It is important to make these two definitions separately because we will be comparing system performance for different $\delta$.  For Nyquist transmission, $\delta=1$ and the two definitions become equal to each other, and a separate definition becomes unnecessary. 

In this paper, we study the {IAPR} behavior of FTN under two types of power configurations, fixed transmit SNR or fixed receive SNR. 
Fixed transmit SNR means that the transmit power $P$ is fixed for all choices of $\delta$. This implies that as $\delta$ gets smaller the received SNR also becomes smaller. Note that this can imply inferior performance, for example, in bit error rate in practical systems. On the other hand, fixed received SNR refers to a fixed symbol energy $E$ for all choices of $\delta$. In other words, for this latter case, we fix $P\delta $ instead of $P$ itself, implying larger transmit power $P$ for smaller $\delta$. However,  increasing $P$ indefinitely is practically impossible due to limitations in linear power amplification.

 According to the distribution of instant power values, we can calculate the probability of instant power exceeding the back-off value. We define this probability as the outage probability.  In this section, we study the {IAPR} behavior of FTN signaling for uniform power allocation derived in Section \ref{sec:uniformpowallo}. Uniform power in both space and frequency imply that $a_k[m]$ is independent of $a_k[n]$ for $m\neq n$. Moreover, the real and imaginary parts of $a_k[m]$, which are denoted as $a_{r,k}[m]$ and $a_{i,k}$ respectively, are i.i.d. as well.

Without loss of generality, we assume that $N=2M+1$ symbols are transmitted and $x_k(t)$ in \eqref{eqn:xt} can be written as 
\begin{align}
    x_k(t)
    &=\sum_{m=-M}^{M}\left(a_{r,k}[m]+ja_{i,k}[m]\right)p(t-m\delta T)\\
    &= x_{r,k}(t)+jx_{i,k}(t),
\end{align}
where $x_{i,k}(t)$ and $x_{j,k}(t)$ are the real and imaginary parts of $x_k(t)$.  In this paper, we define the {IAPR} as 
\begin{equation}
    \text{{IAPR}}=\frac{|x_k(t)|^2}{P_k}.\label{eqn:defpapr}
\end{equation}
As in  \cite{zhang2022faster}, we can easily see that the signal $x_k(t)$ is a cyclostationary random process with period $\delta T$. 
In this paper, we assume that $N$ is always large enough so that the process $x_k(t)$ is a cyclostationary process. Therefore, with large enough $N$, it is sufficient to study the statistical distribution of power for each $t$  within only one period. We limit our time index $t$ to $[0,\delta T)$. For each time instant $t$, the distribution is different because the coefficients $p(t-m\delta T)$ are time-varying.

The complementary cumulative distribution function (CCDF) represents the probability that a random variable exceeds a specific threshold, providing crucial insights into its tail distribution. 
The CCDF of the instantaneous  power $|x_k(t)|^2$ with respect to $t$ can be defined as 
\begin{equation}
    \mathcal{C}(\gamma;t)=\text{Pr}\left[|x_k(t)|^2\geq\gamma\right]. \label{eqn:defccdf}
\end{equation}
The expression \eqref{eqn:defccdf} is still a function of time $t$. It is more important to investigate the average behavior of CCDF distribution within one period, thus we take the time average of it and define the average CCDF as  
\begin{equation}
    \bar{\mathcal{C}}(\gamma)=\frac{1}{\delta T}\int_0^{\delta T}\mathcal{C}(\gamma;t)dt. \label{eqn:aveccdfdef}
\end{equation}
Utilizing the techniques in \cite[Appendix]{exactpapr}, the expression of average CCDF can be computed as 
\begin{equation}
    \bar{\mathcal{C}}(\gamma)=1-\frac{1}{\delta T} \int_0^{\delta T}\sqrt{\gamma}\int_0^\infty D(\zeta;t)J_1(\sqrt{\gamma}\zeta)d\zeta dt, \label{eqn:aveccdfexp}
\end{equation}
where $J_1(\cdot)$ is the first type Bessel function of order 1, and $D(\zeta;t)$ is written as 
\begin{equation}
    D(\zeta;t)=\frac{1}{2\pi}\int_0^{2\pi}\Phi(\zeta\cos\phi,\zeta\sin\phi;t)d\phi. \label{eqn:intchar}
\end{equation}
In \eqref{eqn:intchar}, $\Phi(u,v;t_0)=\mathbb{E}\left[e^{j(ux_{r,k}(t_0)+vx_{i,k}(t_0))}\right]$ is the joint characteristic function of the real part $x_{r,k}(t)$ and the imaginary part $x_{i,k}(t)$ of the process $x_k(t)$ evaluated at $t_0$. 
\begin{remark}
\label{rem:inspowccdfconvert}
    The expression in \eqref{eqn:aveccdfexp} is about the distribution of instant power, according to the definition of {IAPR} in \eqref{eqn:defpapr}. 
We can simply set $\gamma$ as $\gamma' P_k$ in \eqref{eqn:aveccdfdef}  to obtain the average CCDF of {IAPR} of FTN signaling, namely,
\begin{align}
    \bar{\mathcal{C}}(\gamma'P_k)&=\frac{1}{\delta T}\int_0^{\delta T}\mathcal{C}(\gamma'P_k;t)dt \notag\\
    &=\frac{1}{\delta T}\int_0^{\delta T}\text{Pr}\left[\frac{|x_k(t)|^2}{P_k}\geq\gamma' \right]dt.
\end{align}
We conclude that replacing  $\gamma$ with $\gamma' P_k$ is merely a scaling operation and the average CCDF of {IAPR} has the same behavior as the average CCDF of instant power.
\end{remark}

We then derive the exact instant power distributions both for FTN signaling with Gaussian symbols and QPSK symbols in the following subsections.   

\subsection{{IAPR} Distribution for Gaussian Symbols}
\label{sec:gausthy}

\subsubsection{Average CCDF of instant power with $\mathsf{SNR_{tx}}$ fixed}
Assume that data symbols $a_k[m]$ follow complex Gaussian distribution, namely, $a_k[m]\backsim\mathcal{CN}(0, P_k\delta T)$. Note that the variance $P_k\delta T$ is also the symbol energy. Moreover, 
 $a_{r,k}[m]$ and  $a_{i,k}[m]$ are i.i.d. real Gaussian random variables with zero mean and variance $P_k\delta T/2$. 
  We know that  $x_k(t)$ is the linear combination of multiple complex Gaussian random variables and thus it is also a complex Gaussian random variable. Meanwhile, $x_{r,k}(t)$ and $x_{i,k}(t)$ are real Gaussian random variables. We then have the following theorem.
\begin{theorem}
    \label{thm:gaustxsnr} 
    If $\mathsf{SNR_{tx}}$ is fixed and Gaussian symbols are used, as $\delta$ approaches zero, the average CCDF of the instant power $\bar{\mathcal{C}}(\gamma)$ in \eqref{eqn:aveccdfexp} does not change with $\delta$ and is equal to  
    \begin{equation}
         \bar{\mathcal{C}}(\gamma)=\exp\left(-\frac{\gamma}{P_k\int_{-\frac{1}{2\delta T}}^{\frac{1}{2\delta T}}G(f)df}\right). \label{eqn:gausthytx}
    \end{equation}
\end{theorem}
\begin{proof}
    The proof is shown in Appendix \ref{prof:gaus}.
\end{proof}

It is easy to see that as $\delta$ decreases, the CCDF does not change with $\delta$, since the integration $\int_{-\frac{1}{2\delta T}}^{\frac{1}{2\delta T}}G(f)df$ does not change with $\delta$ for $\delta<\frac{1}{1+\beta}$. In other words, the integration range will be larger than the support of  $G(f)$, which is $[-\frac{1+\beta}{2T}, \frac{1+\beta}{2T}]$. 
\subsubsection{Average CCDF of instant power with $\mathsf{SNR_{rx}}$ fixed}
When $\mathsf{SNR_{rx}}$ is fixed, we replace $P_k \delta T$ with $E$, which is a constant with respect to $\delta$. The behavior of the average CCDF of instant power will be different from the $\mathsf{SNR_{tx}}$ fixed case. The above equation \eqref{eqn:gausthytx} becomes
\begin{equation}
     \bar{\mathcal{C}}(\gamma)=\exp\left(-\frac{\gamma}{\frac{E}{\delta T}\int_{-\frac{1}{2\delta T}}^{\frac{1}{2\delta T}}G(f)df}\right). \label{eqn:gausrxsnrthy}
\end{equation}
As $\delta$ goes to zero, the average CCDF $\bar{\mathcal{C}}(\gamma)$ of instant power approaches 1 asymptotically, this means that the average CCDF curve for instant power for fixed $\mathsf{SNR_{rx}}$ with Gaussian symbols approaches a horizontal line as $\delta \rightarrow 0$. 

\begin{remark}
    The equations  \eqref{eqn:gausthytx} and \eqref{eqn:gausrxsnrthy} are average CCDF for instant power. According to Remark \ref{rem:inspowccdfconvert}, 
    we conclude that the average CCDF of {IAPR} for FTN with Gaussian symbols either for $\mathsf{SNR_{tx}}$ or $\mathsf{SNR_{rx}}$  has the same behavior as the average CCDF of instant power.
\end{remark}

\vspace{-0.2cm}
\subsection{{IAPR} Distribution with QPSK Symbol Set}
\label{sec:thyqpsk}

Gaussian signaling is relevant to theoretical results. In practice, we also need to investigate the {IAPR} behavior for practical constellations such as PSK or QAM. 
For simplicity, in this paper we will study the QPSK symbol set. The analysis can be extended to higher-order PSK or QAM constellations similarly.

\subsubsection{Average CCDF of instant power with $\mathsf{SNR_{tx}}$ fixed}
\label{sec:qpsktxsnrccdf}

 In this section, we adopt the analysis methodology in \cite{exactpapr} to find the distribution of instant power of FTN signaling.  We assume that the data symbols $a_k[m]$ are i.i.d. and the constellation has energy $P_k\delta T$ for fixed $\mathsf{SNR_{tx}}$, where $k$ is the antenna index. For fixed transmit SNR, the physical transmission power of FTN is the same for all $\delta$, including Nyquist transmission for $\delta=1$. This means that as $\delta$ decreases from 1 to 0, the physical transmission power does not change, but the average constellation energy decreases.
Similar to Gaussian signaling, the constellation points $a_k[m]=a_{r,k}[m]+ja_{i,k}[m]$ are composed of real and imaginary parts, which are independent of each other\footnote{However, this is not necessarily true for all constellations, for example, in high-order PSK constellations, the real part and the imaginary part of the symbol are highly correlated.} and are drawn uniformly from the set,  
$a_{r,k}[m], a_{i,k}[m]\in\mathcal{A}=\left\{+\sqrt{P_k\delta T/2}, -\sqrt{P_k\delta T/2}\right\}$ for QPSK transmission. Then we have the following theorem.

\begin{theorem}
\label{thm:txsnrpaprqpsk}
    If $\mathsf{SNR_{tx}}$ is fixed and the QPSK symbol set is used, as $\delta$ approaches $0$, the average CCDF of instant power $\bar{\mathcal{C}}(\gamma)$ in \eqref{eqn:aveccdfexp}  asymptotically approaches $1$. Namely,
    \begin{equation}
        \underset{\delta\rightarrow0}{\lim}\bar{\mathcal{C}}(\gamma)=1.
    \end{equation}
\end{theorem}
\begin{proof}
    The proof is provided in Appendix \ref{app:proofthmpaprtxqpsk}.
\end{proof}

 In other words, as $\delta$ decreases, the instant power of FTN transmission takes larger values. In the limit, the probability density function of the instant power distribution of $x_k(t)$ will approach a Dirac delta function located at infinity. On the other hand, since the average power $P_k$ is fixed, according to Remark \ref{rem:inspowccdfconvert}, the distribution of {IAPR} $|x(t)|^2/P_k$ follows the same behavior as the distribution of instant power. 

\subsubsection{Average CCDF of instant power with $\mathsf{SNR_{rx}}$ fixed}

We now investigate the behavior of the average instant power CCDF for {IAPR} when the received power is fixed. In the received SNR fixed scenario, the symbol energy $E$ is kept the same for all $\delta$, namely, the product $P_k\delta$ is fixed regardless of the value of $\delta$. We then have the following theorem.
\begin{theorem}
\label{thm:rxsnrqpsk}
    If $\mathsf{SNR_{rx}}$ is fixed and the QPSK symbol set is used, as $\delta$ approaches $0$, the average CCDF of instant power $\bar{\mathcal{C}}(\gamma)$ asymptotically approaches $1$ as well. We have
    \begin{equation}
        \underset{\delta\rightarrow0}{\lim}\bar{\mathcal{C}}(\gamma)=1.
    \end{equation}
\end{theorem}
\begin{proof}
    The proof is provided in Appendix \ref{prof:qpskrxsnr}.
\end{proof}
To calculate the  {IAPR} in this scenario, we use the same definition as in \eqref{eqn:defccdf}.  We then apply a similar discussion as we did for transmit SNR fixed.  We replace $\gamma$ with $\gamma' P_k$ in \eqref{eqn:defccdf}, and observe that the average CCDF of {IAPR} of FTN transmission for fixed $\mathsf{SNR_{rx}}$ has the same behavior as its $\mathsf{SNR_{tx}}$ fixed counterpart. 

\begin{remark}
Assume $PT = E$ for all $\delta$. As \(\delta\) decreases, the symbol energy in the \(\mathsf{SNR_{rx}}\) fixed case remains constant at \(E\), whereas the symbol energy in the \(\mathsf{SNR_{tx}}\) fixed case, \(P_k \delta T\), decreases because \(P_k\) remains constant. When comparing their physical power, the \(\mathsf{SNR_{rx}}\) fixed FTN has higher physical power, given by \(\frac{E}{\delta T}\), while the \(\mathsf{SNR_{tx}}\) fixed FTN maintains a constant physical power proportional to \(P_k\). Examining the CCDF of instantaneous power, we observe that for the same threshold \(\gamma\), the \(\mathsf{SNR_{rx}}\) fixed FTN is more likely to exhibit instantaneous power values larger than \(\gamma\), due to its higher symbol energy. Consequently, the CCDF curve for \(\mathsf{SNR_{rx}}\) fixed FTN lies above that of \(\mathsf{SNR_{tx}}\) fixed FTN and approaches the horizontal line faster as \(\delta \to 0\). This shows that keeping the symbol energy constant for {decreasing acceleration factor} deteriorates the {IAPR} performance faster than maintaining constant physical transmission power.
\end{remark}
{\begin{remark}
    The IAPR distribution expression for the QPSK symbol set can be extended to other constellation sets by changing the joint characteristic function $\Phi(\zeta\cos\phi,\zeta\sin\phi;t)$ of \eqref{eqn:intchar} according to the constellation. 
\end{remark}}

\vspace{-0.2cm}
\subsection{Asymptotic Behavior of FTN Signaling}

In Theorem \ref{thm:gaustxsnr}, we learned that the average CCDF of {IAPR} of FTN signaling for fixed $\mathsf{SNR_{tx}}$ with Gaussian symbols does not change with $\delta$ as $\delta$ goes to zero. Meanwhile, the average CCDF of {IAPR} of FTN signaling for fixed $\mathsf{SNR_{tx}}$ with QPSK symbols behaves differently. The average CCDF approaches 1 as $\delta\rightarrow0$ as shown in Theorem~\ref{thm:txsnrpaprqpsk}.  
Intuitively, as $\delta$ approaches 0, the signal $x_k(t)$ with QPSK symbols approaches
\begin{equation}
    x_k(t)\approx\left(\sum_{m=-M}^{M}a_k[m]\right)p(t), \label{eqn:llnequ}
\end{equation}
due to the fact that the transmitted symbols are highly packed and the signal starts to look like one-symbol transmission with all the transmitted symbols transmitted at once. One can get the impression that the process in \eqref{eqn:llnequ}  starts to look like a Gaussian process, since we can invoke the law of large numbers.  However, this is not the case if the number of symbols is large enough and consequently the process $x_k(t)$ is always a cyclostationary process. 
We have the following theorem.
\begin{theorem}
\label{thm:qpsknogaus}
    FTN signaling with QPSK symbols do not approach the Gaussian process for arbitrary small non-zero $\delta$, as long as $S$ is sufficiently large. 
\end{theorem}
\begin{proof}
    The proof is provided in Appendix \ref{app:asymgqpsknogaus}
\end{proof}

\begin{remark}
    Since the average CCDF of instant power has the same behavior as the average CCDF of {IAPR}, we can infer from Theorem \ref{thm:qpsknogaus} that as $\delta$ keeps decreasing, the average CCDF of {IAPR} for FTN using QPSK symbols with $\mathsf{SNR_{tx}}$ fixed will keep increasing and will not converge to the average CCDF of {IAPR} using Gaussian symbols with $\mathsf{SNR_{tx}}$ fixed.  
\end{remark}

\begin{figure}[t]
\centering
    \includegraphics[scale=0.47]{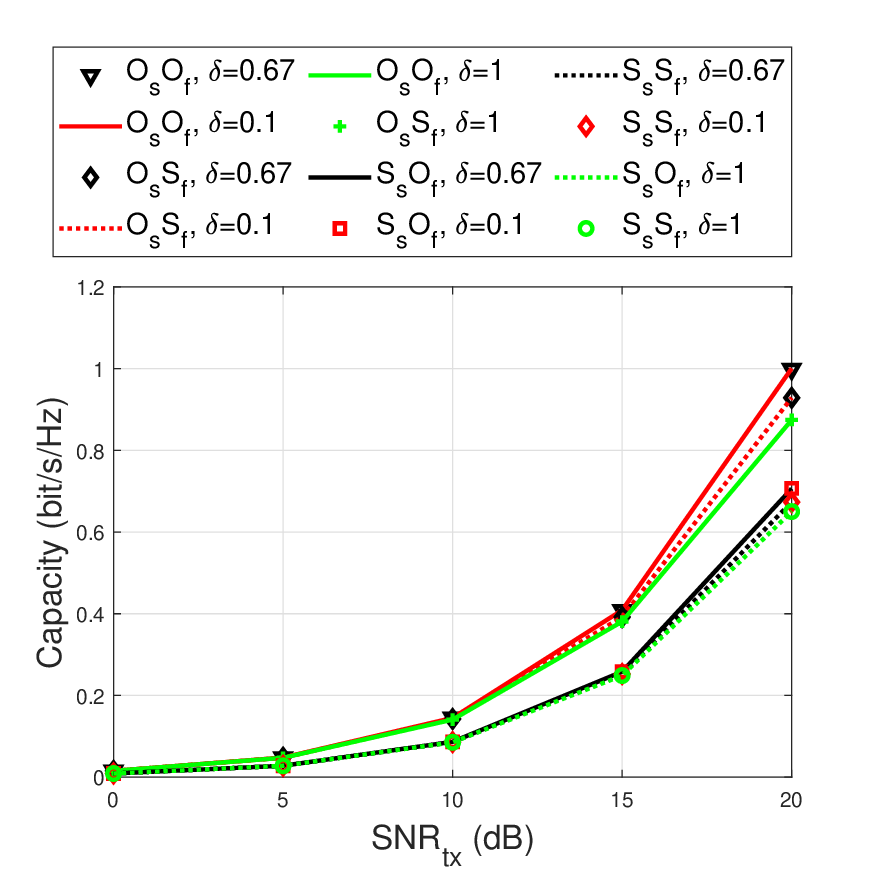}
    \caption{Capacity vs transmit SNR for four different power allocation schemes, $O_sO_f, S_sO_f, O_sS_f$, and $S_sS_f$, with different $\delta$ values. The MIMO size is $2\times 2$.}
    \label{fig:txsnrvsc}
    \centering
    \includegraphics[scale=0.47]{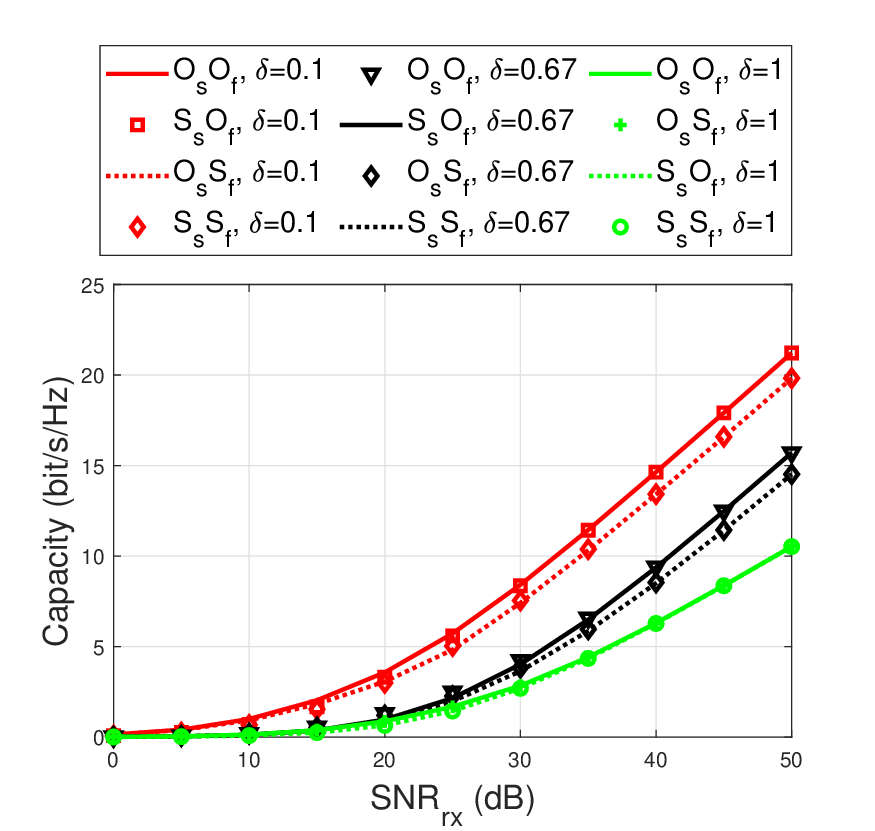}
    \caption{Capacity vs receive SNR for four different power allocation schemes, $O_sO_f, S_sO_f, O_sS_f$, and $S_sS_f$, with different $\delta$ values. The MIMO size is $2\times 2$.}
    \label{fig:rxsnrvsc}
    
\end{figure}

\section{Simulation Results}\label{sec:simresult}

In this section, we show the capacity and {IAPR} of MIMO FTN signaling under different power allocation schemes with either transmit or receive power constraints. 

 In the simulations conducted in this section, we set the symbol period $T=0.01$. The simulation results are averaged over 1000 random channel realizations. In our simulations, the MIMO channel coefficients $h_{l,k}$ are i.i.d. and complex Gaussian distributed according to  $\mathcal{CN}\left(0,\frac{1}{K}\right)$. We assume RRC pulse shaping with roll-off factor $\beta=0.5$. Note that for this $\beta$, the threshold $\delta=0.67$. For the capacity results in Figs. \ref{fig:txsnrvsc}-\ref{fig:tauvscrxsnr}, we apply the power allocation schemes $O_sO_f$, $O_sS_f$, $S_sO_f$, and $S_sS_f$ described in Section \ref{sec:capacityderivation}. For the {IAPR} simulations in Figs.~\ref{fig:gaustxsnr}-\ref{fig:figBqpsk}, we transmit $1000$ symbols.

\begin{figure}[t]
    \centering
    \includegraphics[scale=0.50]{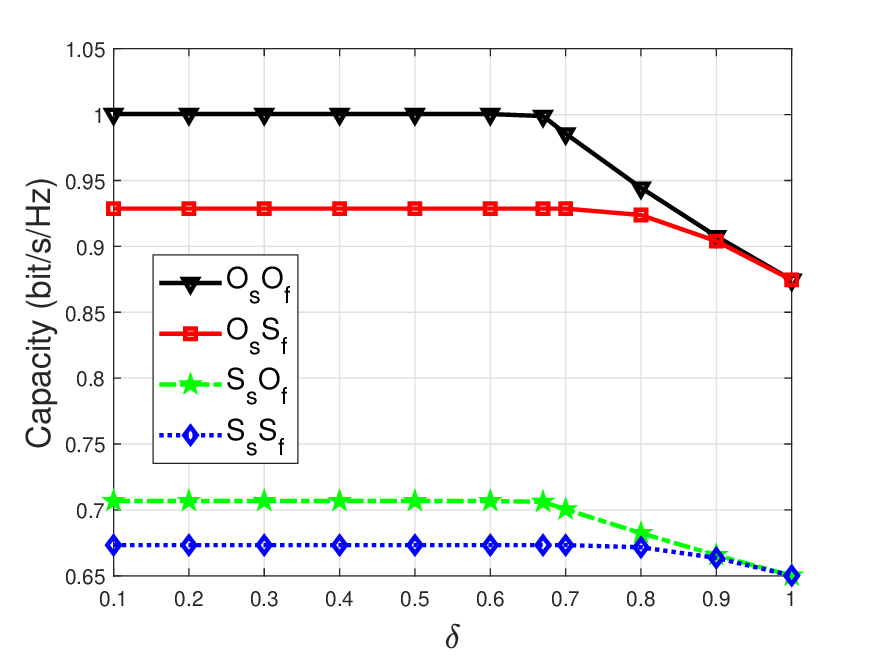}
    \caption{Capacity vs $\delta$ for four different power allocation schemes with transmit  SNR $\mathsf{SNR_{tx}}=20 $ dB. The MIMO size is $2\times 2$ with threshold $\delta$ value to be 0.67.}
    \label{fig:tauvsctxsnr}
\end{figure}

Fig.~\ref{fig:txsnrvsc} shows the MIMO FTN capacity versus $\mathsf{SNR_{tx}}$ for various $\delta$ values and for four power allocation schemes defined in \eqref{eqn:optsolu}, \eqref{eqn:ssof}, \eqref{eqn:ossf}, and \eqref{eqn:sssf}. For $\delta=0.67$ and $\delta=0.1$, the curves overlap as predicted by Theorem \ref{thm:them1}, since MIMO FTN capacity saturates below $\delta = \frac{1}{1+\beta}$. At $\delta=1$, the curves for $O_sO_f$ and $S_sO_f$ align with $O_sS_f$ and $S_sS_f$, as channel inversion yields the same result as uniform power allocation for Nyquist transmission. Compared to Nyquist signaling with uniform power allocation ($\delta=1$ and $S_sS_f$), FTN signaling with optimal power allocation ($\delta=0.67$ and $O_sO_f$) boosts rates by about $50\%$. However, uniform power allocation in frequency still performs well, as seen by comparing $O_sO_f$ and $O_sS_f$ at $\delta=0.67$. Therefore, in practice, uniform power allocation in the frequency domain can still be applied to avoid the extreme values brought by spectrum inversion.


Fig.~\ref{fig:rxsnrvsc} examines MIMO FTN capacity versus $\mathsf{SNR_{rx}}$. The $\delta=1$ curves show the worst performance, as smaller $\delta$ enables higher symbol rates despite equal symbol power. At high SNR, $O_sO_f$ and $O_sS_f$ curves converge with $S_sO_f$ and $S_sS_f$, as water-filling's impact diminishes. FTN transmission with optimal power allocation significantly outperforms Nyquist signaling ($\delta=1$), with $O_sO_f$ and $O_sS_f$ at $\delta=0.67$ achieving a $50\%$ improvement. The advantage is even greater for fixed receive SNR, where higher transmission rates amplify the benefits of FTN while the received symbol energy is kept the same. However, unlike in Fig.~\ref{fig:txsnrvsc}, where $\delta=0.1$ and $\delta=0.67$ curves overlap for all power allocation schemes, this is not observed in Fig.~\ref{fig:rxsnrvsc}. For fixed  $\mathsf{SNR_{rx}}$, although the symbol energy $P\delta T$ remains constant across accelerations, smaller $\delta$ requires higher physical transmission power $P$ to maintain the same symbol energy, amplifying the benefits of FTN.

\begin{figure}[t]
    \centering
    \includegraphics[scale=0.50]{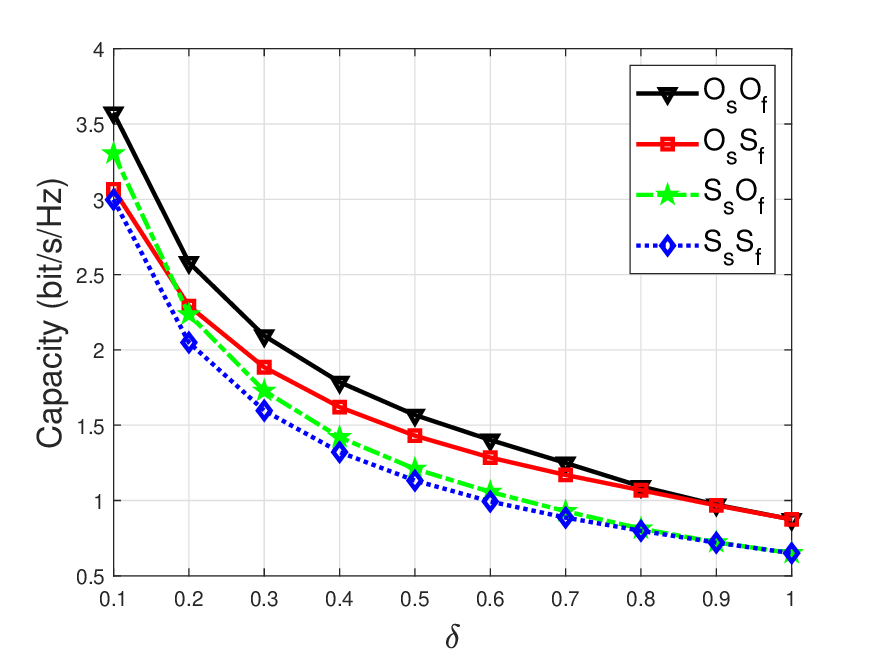}
    \caption{Capacity vs $\delta$ for four different power allocation schemes with receive  SNR $\mathsf{SNR_{rx}}=20 $ dB. The MIMO size is $2\times 2$ with threshold $\delta$ value to be 0.67.}
    \label{fig:tauvscrxsnr}
\end{figure}

In Fig. \ref{fig:tauvsctxsnr}, we show the capacity of MIMO FTN with respect to $\delta$ for fixed transmission power.  As we explained in Section \ref{sec:capacityderivation}, the capacity increases as $\delta$ decreases until $\delta$ reaches $\frac{1}{1+\beta}$, then it remains fixed. We can also see that the curves with optimal frequency domain power allocation scheme $O_f$ grow faster as $\delta$ decreases. 
We also notice that spatial domain water-filling $O_s$ provides a more significant gain than frequency domain inversion over the support.  Meanwhile, in Fig.~\ref{fig:tauvscrxsnr} we again plot the capacity of MIMO FTN with respect to $\delta$ but with fixed $\mathsf{SNR_{rx}}$ instead. In this case, as $\delta$ decreases, capacity keeps increasing. Moreover, as $\mathsf{SNR_{rx}}$ increases the improvement brought by spatial domain water-filling, $O_s$, is less significant.

\begin{figure}
    \centering
    \includegraphics[width=0.9\linewidth]{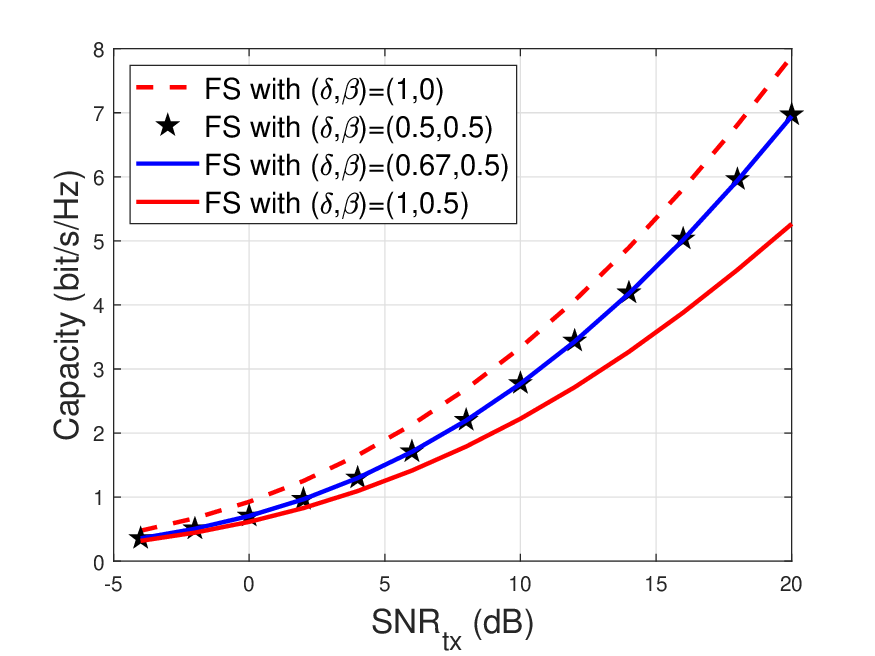}
    \caption{Capacity of MIMO FTN in FS channels vs $\mathsf{SNR_{tx}}$ for different $(\delta,\beta)$ pairs, where the MIMO size is $2\times2$.}
    \label{fig:fsfccap}
\end{figure}



{Note that, whenever FTN is used, optimal power allocation in the frequency domain will always be the best. However, due to the amount of computation required for shaping the power spectrum of the input, in practice, FTN signaling with i.i.d. input distribution can still perform well. On the other hand, although water-filling in the spatial domain also improves the capacity compared to uniform power allocation, channel state information (CSI) is not always available at the transmitter, or we may have imperfect CSI. Imperfect CSI due to estimation errors or outdated channel information leads to capacity degradation and reduced eigenchannel utilization. In practice, if CSI is unavailable or inaccurate, uniform power allocation in the spatial domain can be used.}

{In Fig.~\ref{fig:fsfccap}, we show the capacity of MIMO FTN in FS channels for different $(\delta,\beta)$ pairs. We assumed the delay taps are in a range of $[0,2T)$. We assume all channel coefficients $h_{kl}^j$ are independent and identically distributed according to the complex Gaussian distribution $\mathcal{CN}(0,1/(KJ))$. We set the number of taps to $J=20$, {$N=1000$}, and MIMO size $2\times2$. We apply the optimal power allocation in \eqref{eqn:fsfcoptpowallo}. We can see that the capacity stops growing after $\delta$ reaches the threshold $\frac{1}{1+\beta}$.   }

To evaluate the practicality of MIMO FTN signaling, we next simulate the {IAPR} performance under various signaling rates and plot the empirical CCDF of the {IAPR} to analyze outage probability and threshold $\gamma$, particularly at {small acceleration factor}. 
The definition of average CCDF of instant power can be found in \eqref{eqn:aveccdfdef}, we then replace $\gamma$ with $\gamma' P_k$ to obtain the average CCDF of {IAPR} as shown in Remark \ref{rem:inspowccdfconvert}.  





Fig.~\ref{fig:gaustxsnr} shows the CCDF of SISO FTN with a Gaussian symbol set at a fixed transmit SNR of $\mathsf{SNR_{tx}}=20$ dB. As discussed in Section \ref{sec:gausthy}, the {IAPR} distribution is unaffected by $\delta$. The theoretical distribution from \eqref{eqn:gausthytx} closely matches the simulation results. Fig.~\ref{fig:gausrxsnr} presents the average CCDF of SISO FTN with Gaussian symbols for fixed $\mathsf{SNR_{tx}}$, showing that decreasing $\delta$ degrades performance at a fixed $\mathsf{SNR_{rx}}$.


\begin{figure}[t]
    \centering
    \includegraphics[scale=0.50]{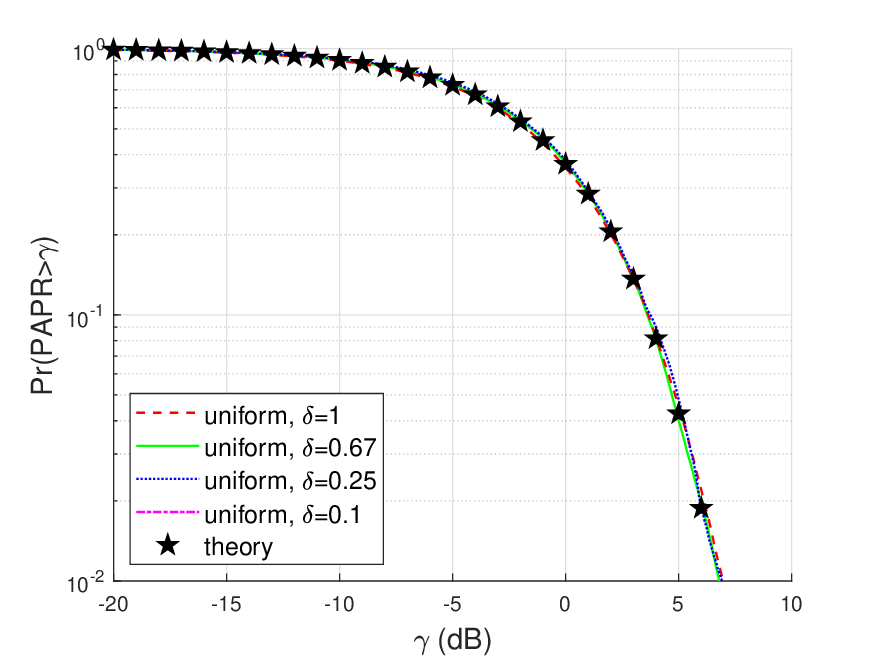}
    \caption{CCDF of SISO FTN for different $\delta$ for only uniform power allocation both in space and in frequency. The Gaussian symbol set is used, $\mathsf{SNR_{tx}}$ is fixed. }
    \label{fig:gaustxsnr}
\end{figure}
\begin{figure}[t]
    \centering
    \includegraphics[scale=0.50]{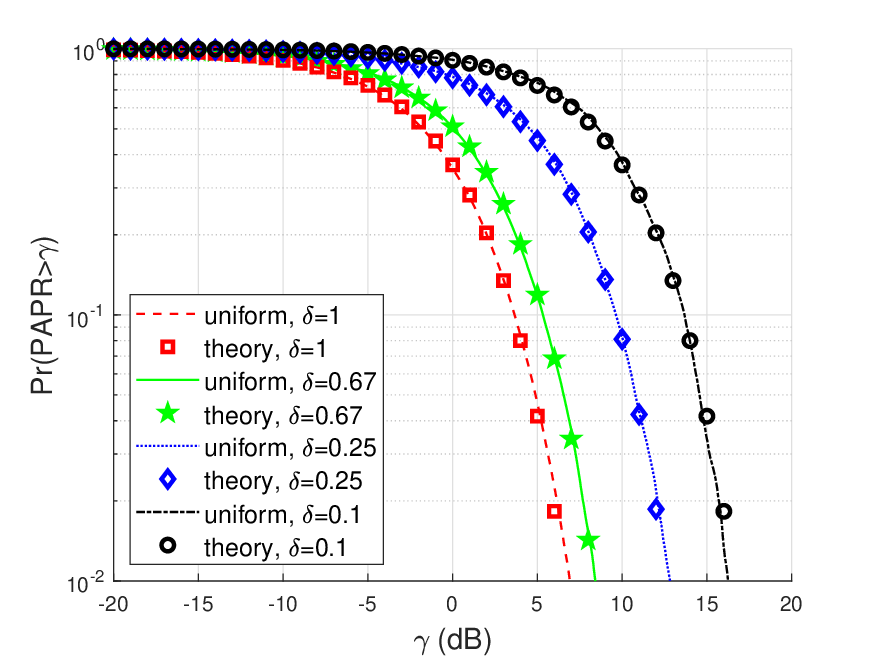}
    \caption{CCDF of SISO FTN for different $\delta$ for only uniform power allocation both in space and in frequency. The Gaussian symbol set is used, $\mathsf{SNR_{rx}}$ is fixed. The theoretical CCDF in \eqref{eqn:gausrxsnrthy} is also plotted. }
    \label{fig:gausrxsnr}
\end{figure}
\begin{figure}[t]
    \centering
    \includegraphics[scale=0.50]{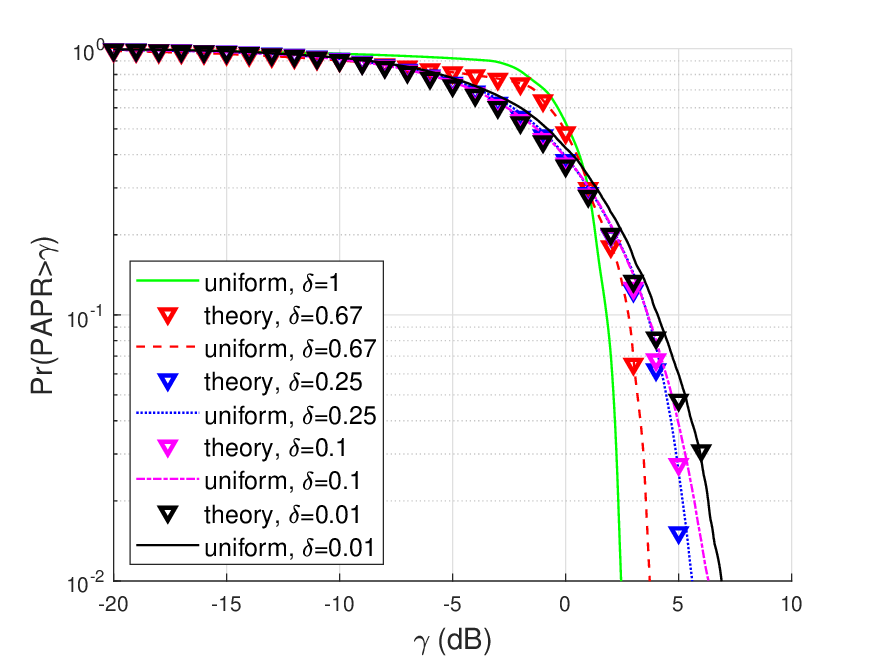}
    \caption{CCDF of SISO FTN with different $\delta$ values for only uniform power allocation both in space and in frequency. The QPSK symbol set is used, $\mathsf{SNR_{tx}}$ is fixed. The theoretical curve in \eqref{eqn:extdist} is also plotted.}
    \label{fig:qpsktxsnr}
\end{figure}
\begin{figure}[t]
    \centering
    \includegraphics[scale=0.50]{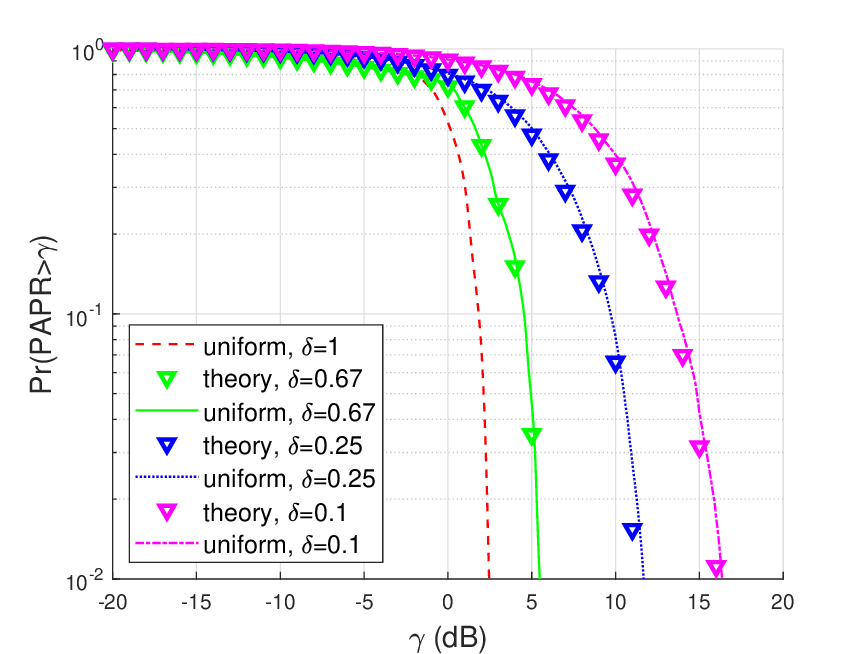}
    \caption{CCDF of SISO FTN for different $\delta$ for only uniform power allocation both in space and in frequency. The QPSK symbol set is used, $\mathsf{SNR_{rx}}$ is fixed. The theoretical CCDF in \eqref{eqn:extdist} is also plotted. }
    \label{fig:qpskrxsnr}
\end{figure}

Fig.~\ref{fig:qpsktxsnr} shows the CCDF for the QPSK symbol set at $\mathsf{SNR_{tx}}=20$ dB, with the theoretical distribution from \eqref{eqn:extdist} aligning with simulation results. The CCDF rises as $\delta$ decreases, consistent with Section \ref{sec:thyqpsk}. Similarly,  Fig.~\ref{fig:qpskrxsnr} illustrates the average CCDF of SISO FTN with QPSK at $\mathsf{SNR_{rx}}=20$ dB, highlighting the impact of different acceleration factors. As $\delta$ decreases, the {IAPR} distribution worsens. Since the symbol power is fixed for different $\delta$ values, {smaller acceleration factor} increases pulse overlap and causes more {IAPR} spread. However, QPSK generally outperforms the Gaussian symbol set in {IAPR} distribution.


\begin{figure}[t]
    \centering
    \includegraphics[scale=0.50]{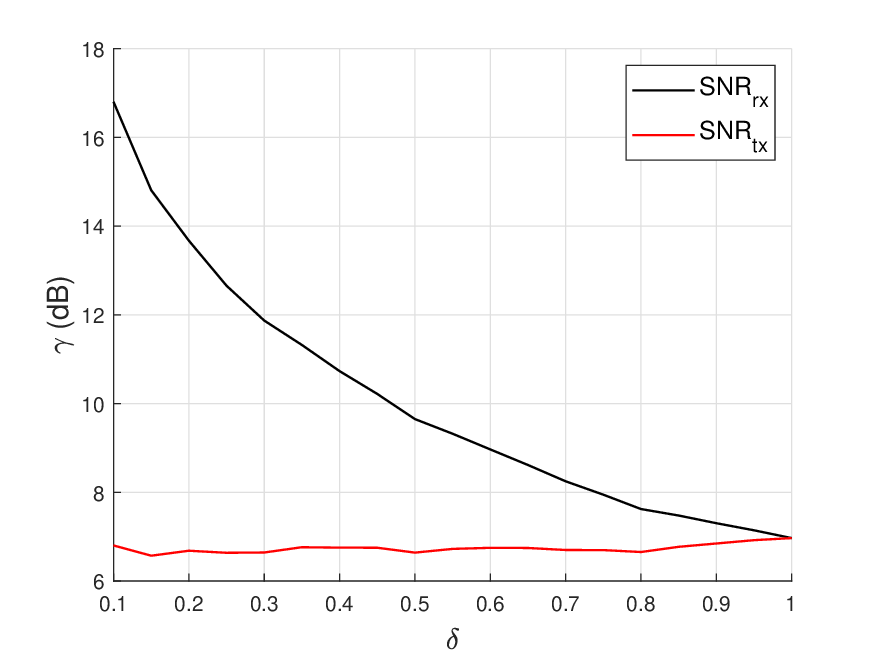}
    \caption{Gamma value versus $\delta$ for  Gaussian symbol set. }
    \label{fig:figBgaus}
\end{figure}
\begin{figure}[t]
    \centering
    \includegraphics[scale=0.50]{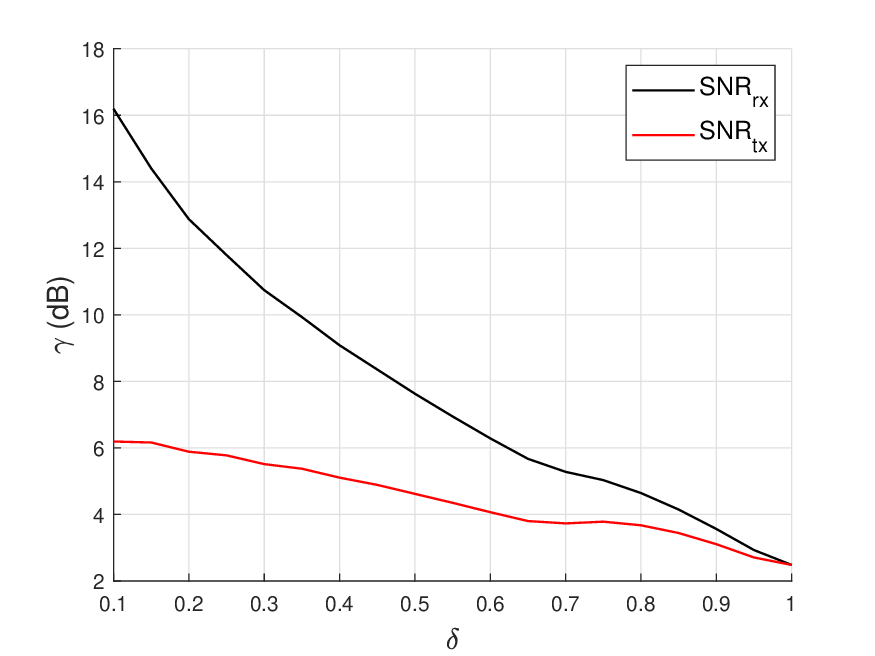}
    \caption{Gamma value versus $\delta$ for  QPSK symbol set. }
    \label{fig:figBqpsk}
\end{figure}

Finally, we simulate the outage threshold $\gamma$ for a fixed outage probability of 0.01. Fig.~\ref{fig:figBgaus} shows the effect of $\delta$ on $\gamma$ in SISO FTN using the Gaussian symbol set with uniform power allocation $S_sS_f$, comparing $\mathsf{SNR_{tx}}=20$~dB and $\mathsf{SNR_{rx}}=20$~dB. With fixed transmit SNR, the outage threshold stays high, while for fixed receive SNR, $\gamma$ increases as $\delta$ decreases, aligning with Fig.~\ref{fig:gausrxsnr}. Fig.~\ref{fig:figBqpsk} presents similar results for the QPSK symbol set, showing that QPSK achieves better performance with smaller $\gamma$ values than the Gaussian symbol set.


Note that we obtain stable {IAPR} results for frequency domain uniform power allocation schemes, $S_f$, in the simulations conducted in Figs. \ref{fig:figBgaus} and \ref{fig:figBqpsk}. The inversion of the frequency spectrum magnifies the close-to-zero values and these values are highly sensitive to $\delta$, number of symbols $N$, and the simulation length.

{Under fixed $\mathsf{SNR_{tx}}$, decreasing the FTN acceleration factor $\delta$ increases IAPR and thus outage probability, requiring greater PA back-off and reducing power efficiency, while the capacity remains unchanged. Therefore, choosing $\delta \approx \frac{1}{1+\beta}$ offers a practical balance between performance and IAPR. In contrast, under fixed $\mathsf{SNR_{rx}}$, decreasing $\delta$ improves capacity but sharply increases IAPR, demanding larger back-off. Thus, a trade-off arises between spectral efficiency and PA efficiency, which must be carefully managed in FTN system design.}

\section{Conclusion}\label{sec:conclusion}
In this paper, we investigate the capacity of FTN signaling for frequency-flat  MIMO
channels for arbitrary {small acceleration factor}. The frequency domain expression is derived for the MIMO FTN capacity calculations. The optimum power allocation scheme is spatial domain water-filling and frequency domain channel inversion for all $\delta$. 
 For RRC pulses with roll-off factor $\beta$, and for fixed transmit power, for $\delta<\frac{1}{1+\beta}$, the capacity is independent of $\delta$ and remains constant. Moreover, we derive the theoretical {IAPR} distribution of FTN for Gaussian and QPSK signaling. The theoretical distributions closely align with the simulation results. We find that as {the acceleration factor decreases}, the {IAPR} grows unbounded for fixed $\mathsf{SNR_{rx}}$ case for both Gaussian and QPSK symbols. It is also unbounded for QPSK for fixed $\mathsf{SNR_{tx}}$.     
 Overall, we show
that MIMO FTN significantly improves spectral efficiency
with respect to Nyquist transmission, and practical FTN designs should take {IAPR} into account. As future work, we plan to perform a comprehensive analysis of BER and spectral efficiency in conjunction with {IAPR} evaluations. The real potential of FTN will reveal itself when different modulation and coding schemes at different acceleration factors are compared under a given {IAPR} constraint. 

\appendices
\section{Proof for Theorem \ref{thm:gaustxsnr}}
\label{prof:gaus}
We can find the variance for $x_{r,k}(t)$ and $x_{i,k}(t)$ as 
\begin{align}
    \mathbb{E}\left[|x_{r,k}(t)|^2\right]&=\sum_{i=-M}^M\mathbb{E}\left[|a_{r,k}[i]|^2\right]p^2_{i,\delta}(t) \notag
    \end{align}
    \begin{align}
    &=\frac{P_k\delta T}{2}\sum_{i=-M}^Mp^2_{i,\delta}(t)\notag\\&=\mathbb{E}\left[|x_{i,k}(t)|^2\right].
\end{align}
The random variable  $|x_k(t)|^2=|x_{r,k}(t)|^2+|x_{i,k}(t)|^2$ follows a scaled chi-squared distribution with two degrees of freedom, which is also a Gamma distribution with parameters $1$ and $P_k\delta T\sum_{i=-M}^Mp^2_i(t)$.
The average CCDF $\bar{\mathcal{C}}(\gamma)$ is written as
\begin{align}
    \bar{\mathcal{C}}(\gamma)=\frac{1}{\delta T}\int_0^{\delta T}e^{-\frac{\gamma}{P_k\delta T\sum_{-M}^Mp^2_m(t)}}dt. \label{eqn:ccdfgaus}
\end{align}

We assume that the energy of $p(t)$ is concentrated in a finite time period. 
Since we assumed that the number of symbols is large enough, we can approximate $\sum_{i=-M}^Mp^2_i(t)$ by $\sum_{i=-\infty}^{+\infty}p^2_i(t)$. Since we focus on the asymptotic behavior of the CCDF as $\delta$ goes to zero, we might as well assume $\delta\leq\frac{1}{1+\beta}$.
\begingroup
\setlength{\belowdisplayskip}{-15pt}
\begin{figure*}
    \begin{align}
    \sum_{n=-\infty}^{\infty}\left|p(t-n\delta T)\right|^2
    &=\int_{-\frac{1}{2}}^{\frac{1}{2}}\left|\frac{1}{\delta T}\sum_{n=-\infty}^{\infty}\sqrt{G\left(\frac{f_n-n}{\delta T}\right)}e^{-j2\pi \left(\frac{f_n-n}{\delta T}\right)t}\right|^2df_n \label{eqn:longpaprbeg}\\ 
    &=\delta T\int_{-\frac{1}{2\delta T}}^{\frac{1}{2\delta T}}\left|\frac{1}{\delta T}\sum_{n=-\infty}^{\infty}\sqrt{G\left(f-\frac{n}{\delta T}\right)}e^{-j2\pi \left(f-\frac{n}{\delta T}\right)t}\right|^2df \\
    &=\frac{1}{\delta T}\int_{-\frac{1}{2\delta T}}^{\frac{1}{2\delta T}}\sum_{n=-\infty}^{\infty}\sum_{m=-\infty}^{\infty}\sqrt{G\left(f-\frac{n}{\delta T}\right)}\sqrt{G\left(f-\frac{m}{\delta T}\right)}e^{-j2\pi \frac{(n-m)}{\delta T}t}df \\
    &=\frac{1}{\delta T}\int_{-\frac{1}{2\delta T}}^{\frac{1}{2\delta T}}\sum_{n=-1}^{1}\sum_{m=-1}^{1}\sqrt{G\left(f-\frac{n}{\delta T}\right)}\sqrt{G\left(f-\frac{m}{\delta T}\right)}e^{-j2\pi \frac{(n-m)}{\delta T}t}df \label{eqn:leaveone}\\
    &=\frac{1}{\delta T}\int_{-\frac{1}{2\delta T}}^{\frac{1}{2\delta T}} \left(G(f)+2\left(\sqrt{G\left(f-\frac{1}{\delta T}\right)G(f)}+\sqrt{G\left(f+\frac{1}{\delta T}\right)G(f)}\right)\cos\left(\frac{2\pi t}{\delta T}\right)\right)df  \label{eqn:longpaprend}
\end{align}
\end{figure*}

\endgroup
We then obtain the derivations in \eqref{eqn:longpaprbeg}-\eqref{eqn:longpaprend}. Here \eqref{eqn:leaveone} is because we use RRC for $p(t)$, and inside the integration interval $[-\frac{1}{2\delta T},\frac{1}{2\delta T}]$, only the closest shifted versions of $G(f)$ will be present; i.e. $m=\pm1$ and $n=\pm1$. Furthermore, if $\delta(1+\beta)\leq 1$, there will only be one spectrum $G(f)$ inside the interval. Then in \eqref{eqn:longpaprend}, as $\delta\leq\frac{1}{1+\beta}$, the term $2\left(\sqrt{G\left(f-\frac{n}{\delta T}\right)G(f)}+\sqrt{G\left(f+\frac{n}{\delta T}\right)G(f)}\right)\cos\left(\frac{2\pi t}{\delta T}\right)$ will disappear in the integration. Therefore, we have
\begin{align}
    \sum_{i=-\infty}^{+\infty}p^2_i(t)&=\sum_{i=-\infty}^{+\infty}|p_i(t)|^2=\frac{1}{\delta T}\int_{-\frac{1}{2\delta T}}^{\frac{1}{2\delta T}}G(f)df, 
\end{align}
and \eqref{eqn:ccdfgaus} becomes
\begin{equation}
    \bar{\mathcal{C}}(\gamma)=\exp\left(-\frac{\gamma}{P_k\int_{-\frac{1}{2\delta T}}^{\frac{1}{2\delta T}}G(f)df}\right). \label{eqn:gausthytxapp}
\end{equation}

\vspace{-0.31cm}

\section{Proof for Theorem \ref{thm:txsnrpaprqpsk}}
\label{app:proofthmpaprtxqpsk}
In this section, we prove that the average CCDF for the process $x_k(t)$ approaches $1$ as $\delta$ approaches $0$ using the QPSK symbol set with $\mathsf{SNR_{tx}}$ fixed.

The random variable $x_k(t_0)=\sum_{m=-M}^{M}a_k[m]p(t_0-m\delta T)$ can be viewed as the linear combination of $2M-1$ independent random variables. For ease of notation, we denote $p(t-m\delta T)$ as $p_{m,\delta}(t)$. Therefore, we can decompose $\Phi(u,v;t_0)$ into 
\begin{equation}
    \Phi(u,v;t_0)=\prod_{m=-M}^M\Phi(p_{m,\delta}(t_0)u,p_{m,\delta}(t_0)v),
\end{equation}
where $\Phi(p_{m,\delta}(t_0)u,p_{m,\delta}(t_0)v)$ is the joint characteristic function of $p_{m,\delta}(t_0){a_{r,k}[m]}$ and $p_{m,\delta}(t_0){a_{i,k}[m]}$. Considering the fact that {$a_{r,k}[m]$} and {$a_{i,k}[m]$} are independent, $\Phi(p_{m,\delta}(t_0)u,p_{m,\delta}(t_0)v)$ can be decomposed into the multiplication of the characteristic functions of $p_{m,\delta}(t_0){a_{r,k}[m]}$ and $p_{m,\delta}(t_0){a_{i,k}[m]}$, namely, 
\begin{equation}
    \Phi(p_{m,\delta}(t_0)u,p_{m,\delta}(t_0)v)=\Phi(p_{m,\delta}(t_0)u)\Phi(p_{m,\delta}(t_0)v).
\end{equation}
Since the symbols $a_k[m]$ are drawn uniformly from the constellation, we know that 
\begin{align}
    \Phi(p_{m,\delta}(t_0)u)
    &=\int_{-\infty}^{+\infty}\left(\frac{1}{2}\delta(\nu-\sqrt{P_k\delta T/2}p_{m,\delta}(t_0)) \right.\notag\\
    &~\quad\quad\quad+\left.\frac{1}{2}\delta(\nu+\sqrt{P_k\delta T/2}p_{m,\delta}(t_0))\right)e^{j\nu u}d\nu\notag \\
    &=\cos(\sqrt{P_k\delta T/2}p_{m,\delta}(t_0)u).
\end{align}
     Similarly, we have
\begin{equation}
    \Phi(p_{m,\delta}(t_0)v)=\cos(\sqrt{P_k\delta T/2}p_{m,\delta}(t_0)v).
\end{equation}
Eventually, we obtain 
\begin{align}
    &\bar{\mathcal{C}}(\gamma)=1-\frac{1}{2\pi\delta T} \int_0^{\delta T}\sqrt{\gamma}\int_0^\infty J_1(\sqrt{\gamma}\zeta) \int_0^{2\pi}\prod_{m=-M}^M \notag\\ &\cos(\sqrt{P_k\delta T/2}p_{m,\delta}(t){\zeta}\cos\phi)\cos(\sqrt{P_k\delta T/2}p_{m,\delta}(t){\zeta}\sin\phi)d\phi d\zeta dt. \label{eqn:extdist}
\end{align}
Note that it is possible that for some $m$ and $t$, $p_{m,\delta}(t)$ equals zero, for the corresponding $m$ and $t$, the term $\cos(\sqrt{P_k\delta T/2}p_{m,\delta}(t)\sqrt{\zeta}\cos\phi)\cos(\sqrt{P_k\delta T/2}p_{m,\delta}(t)\sqrt{\zeta}\sin\phi)$ becomes $1$. However, because of the {small acceleration factor}, other values of $m$ and $t$ will still lead to non-zero $p_{m,\delta}(t)$. Therefore for those $m$ and $t$ that $p_{m,\delta}(t)=0$, their contribution to the whole product $\prod_{m=-M}^M$ $\cos(\sqrt{P_k\delta T/2}p_{m,\delta}(t){\zeta}\cos\phi)\cos(\sqrt{P_k\delta T/2}p_{m,\delta}(t){\zeta}\sin\phi)$ is $1$. Meanwhile, if $\delta=1$, it is possible that for some $m$ and $t$ there is only one term left in the product, for example, when $t=0$. 

We now proceed to discuss the behavior of the distribution at extreme $\delta$ values. Assuming that $\delta$ is sufficiently small,  we are able to approximate the product $\prod_{m=-M}^M$ $\cos(\sqrt{P_k\delta T/2}p_{m,\delta}(t){\zeta}\cos\phi)\cos(\sqrt{P_k\delta T/2}p_{m,\delta}(t){\zeta}\sin\phi)$ with its Taylor expansion. For ease of notation, we denote 
\begin{align}
   c_0= \sqrt{P_k\delta T/2}{\zeta}\cos\phi ,\quad
   c_1=\sqrt{P_k\delta T/2}{\zeta}\sin\phi. \label{eqn:c0c1}
\end{align} We define the index set $\mathcal{I}=\{-M,-M+1, \dots, M-1, \allowbreak M\}$. We use Taylor expansion about $t$ around point $0$ and approximate the multiplication $\Pi_{m=-M}^M\allowbreak\cos(c_0p_{m,\delta}(t))\cos(c_1p_{m,\delta}(t))$ by keeping only the zero-order and the first-order items.  Eventually, we get the approximation in \eqref{eqn:approx}.
\begingroup\setlength{\belowdisplayskip}{-10pt}
\begin{figure*}
    \begin{align}
        &\prod_{m=-M}^M\cos(c_0p_{m,\delta}(t))\cos(c_1p_{m,\delta}(t)) \notag\\
        &\approx \prod_{m=-M}^M\cos(c_0p_{m,\delta}(0))\cos(c_1p_{m,\delta}(0)) - t \bigg(\sum_{m=-M}^Mc_0p'_m(0)\sin(c_0p_{m,\delta}(0))\prod_{n\in\mathcal{I}/m}\cos(c_0p_n(0))\prod_{n=-M}^M\cos(c_1p_n(0)) \notag \\
        & \quad\quad\quad\quad\quad\quad\quad\quad\quad\quad\quad\quad+\sum_{m=-M}^Mc_0p'_m(0)\sin(c_1p_{m,\delta}(0))\prod_{n=-M}^M\cos(c_0p_n(0))\prod_{n\in\mathcal{I}/m}\cos(c_1p_n(0))\bigg)\label{eqn:approx}
    \end{align}
\end{figure*}
\endgroup
The result becomes a linear function of $t$, we plug \eqref{eqn:approx} back into \eqref{eqn:extdist} and get the approximation of average CCDF in \eqref{eqn:distapprox}, where the $\frac{\delta T}{2}$ term comes from the integration $\int_0^{\delta T}tdt$.
\begingroup
\setlength{\belowdisplayskip}{-11pt}
\begin{figure*}
    \begin{align}
    \hline
        &\bar{\mathcal{C}}(\gamma)\approx1-\frac{\sqrt{\gamma}}{2\pi} \int_0^\infty J_1(\sqrt{\gamma}\zeta)\int_0^{2\pi} \prod_{m=-M}^M\cos(c_0p_{m,\delta}(0))\cos(c_1p_{m,\delta}(0))-\frac{\delta T}{2}\bigg(\sum_{m=-M}^Mc_0p'_m(0)\sin(c_0p_{m,\delta}(0))\times \notag\\
        &\prod_{n\in\mathcal{I}/m}\cos(c_0p_n(0))\prod_{n=-M}^M\cos(c_1p_n(0)) +\sum_{m=-M}^Mc_1p'_m(0)\sin(c_1p_{m,\delta}(0))\prod_{n=-M}^M\cos(c_0p_n(0))\prod_{n\in\mathcal{I}/m}\cos(c_1p_n(0))\bigg)d\phi d\zeta \label{eqn:distapprox}
    \end{align}
\end{figure*}
\endgroup
This form allows us to investigate the asymptotic behavior, as $\delta$ approaches zero, the cosine and sine terms approach $1$ and $0$ respectively. As $\int_0^\infty J_1(x)dx=0$, the average CCDF $\bar{\mathcal{C}}(\gamma)$ asymptotically becomes
\begin{align}
    \underset{\delta\rightarrow0}{\lim}\bar{\mathcal{C}}(\gamma)=1- \sqrt{\gamma}\int_0^\infty J_1(\sqrt{\gamma}\zeta)d\zeta =1.
\end{align}

\vspace{-0.2cm}

\section{Proof for Theorem \ref{thm:rxsnrqpsk}}
\label{prof:qpskrxsnr}

{As $\delta$ approaches zero, \eqref{eqn:distapprox}, which is shown on the next page, becomes
\begin{align}
    &\bar{\mathcal{C}}(\gamma)\approx1-\frac{\sqrt{\gamma}}{2\pi} \int_0^\infty J_1(\sqrt{\gamma}\zeta)\int_0^{2\pi} \prod_{m=-M}^M\cos(c_0p_{m,\delta}(0)) \notag\\
    &\quad\quad\times\cos(c_1p_{m,\delta}(0))d\phi d\zeta. \label{eqn:rxsnrapproxpre}
\end{align}
Since the product $P_k\delta$ is fixed, we can replace $P_k\delta$ with $P_\delta$. Then we plug \eqref{eqn:c0c1} into \eqref{eqn:rxsnrapproxpre} and get 
\begin{align}
    &\bar{\mathcal{C}}(\gamma)\approx1-\frac{\sqrt{\gamma}}{2\pi} \int_0^\infty J_1(\sqrt{\gamma}\zeta)\int_0^{2\pi} \prod_{m=-M}^M\cos(\sqrt{P_\delta T/2} \notag\\
    &p_{m,\delta}(0){\zeta}\cos\phi)\cos(\sqrt{P_\delta T/2}p_{m,\delta}(0){\zeta}\sin\phi)d\phi d\zeta. \label{eqn:rxsnrapprox}
\end{align}
}
As $\delta$ becomes small enough, the samples $p_{m,\delta}(t)$ can be considered equal to each other. Without loss of generality, we assume they are equal to $p_0(0)$. Then we have 
\begin{align}
    &\bar{\mathcal{C}}(\gamma)\approx1-\frac{\sqrt{\gamma}}{2\pi} \int_0^\infty J_1(\sqrt{\gamma}\zeta)\int_0^{2\pi}\cos^{2M+1}(c\cos\phi)  \notag\\
    &\quad\quad\quad\quad\quad\quad\quad\quad\quad\quad\times\cos^{2M+1}(c\sin\phi)d\phi d\zeta, \label{eqn:beforeexpandcos} 
\end{align}
where we denote $\sqrt{P_\delta T/2}p_0(0){\zeta}$ as $c$. 
{By applying trigonometric identities, $\cos^{2M+1}(c\cos\phi)$ can be expanded as in \eqref{eqn:cosexpand}.}
\begin{figure*}

    \begin{align}
    \hline
        \cos^{2M+1}(c\cos\phi)&=\frac{1}{2^{2M}}\big[\cos(c\cos\phi+c\cos\phi+\dots+c\cos\phi)+\cos(c\cos\phi+c\cos\phi+\dots-c\cos\phi)+\dots\notag\\
        &\quad\quad\quad\quad\quad\quad+\cos(c\cos\phi-c\cos\phi+\dots-c\cos\phi)\big] \notag\\
        &=\frac{1}{2^{2M}}\sum_{i=0}^{2M}\left(\begin{matrix}
        2M\\i
    \end{matrix}\right)\cos((2M+1-2i)c\cos\phi) \label{eqn:cosexpand}
    \end{align}
\end{figure*}
{Based on the expansion of \eqref{eqn:cosexpand},}
\begingroup
\setlength{\belowdisplayskip}{-20pt}
\begin{figure*}

    \begin{align}
    \hline
    \cos^{2M+1}(c\cos\phi)\cos^{2M+1}(c\sin\phi)&=\left(\frac{1}{2^{2M}}\sum_{i=0}^{2M}\left(\begin{matrix}
        2M\\i
    \end{matrix}\right)\cos((2M+1-2i)c\cos\phi\right)\left(\frac{1}{2^{2M}}\sum_{j=0}^{2M}\left(\begin{matrix}
        2M\\j
    \end{matrix}\right)\cos((2M+1-2j)c\sin\phi\right) \label{eqn:rxsnrcosexpandstart}\\
    &=\frac{1}{2^{4M}}\sum_{i=0}^{2M}\sum_{j=0}^{2M}\left(\begin{matrix}
        2M\\i
    \end{matrix}\right)\left(\begin{matrix}
        2M\\j
    \end{matrix}\right)\cos((2M+1-2i)c\cos\phi) 
 \cos((2M+1-2j)c\sin\phi) \notag\\
 &=\frac{1}{2^{4M+1}}\sum_{i=0}^{2M}\sum_{j=0}^{2M}\left(\begin{matrix}
        2M\\i
    \end{matrix}\right)\left(\begin{matrix}
        2M\\j
    \end{matrix}\right)\left(\cos(d_{i,j}c\cos(\phi+\theta_{i,j})) + 
 \cos(d_{i,j}c\cos(\phi-\theta_{i,j}))\right)\label{eqn:rxsnrcosexpandend}
\end{align}
\end{figure*}
\endgroup
{we obtain the derivations in \eqref{eqn:rxsnrcosexpandstart}-\eqref{eqn:rxsnrcosexpandend}, where $d_{i,j}=\sqrt{(2M+1-2i)^2+(2M+1-2j)^2}$ and $\cos\theta_{i,j}=\frac{2M+1-2i}{\sqrt{(2M+1-2i)^2+(2M+1-2j)^2}}$.} We then integrate \eqref{eqn:rxsnrcosexpandend} over $[0, 2\pi)$. Note that this integration is the first kind of Bessel function of order zero, namely,
\begin{align}
    \frac{1}{2\pi}\int_0^{2\pi}\cos(d_{k,l}c\cos(\phi+\theta_{k,l}))d\phi=J_0(d_{k,l}c). \label{eqn:int1}
\end{align}
Similarly, we have 
\begin{align}
    \frac{1}{2\pi}\int_0^{2\pi}\cos(d_{k,l}c\cos(\phi-\theta_{k,l}))d\phi=J_0(d_{k,l}c). \label{eqn:int2}
\end{align}
We plug \eqref{eqn:int1} and \eqref{eqn:int2} back in \eqref{eqn:beforeexpandcos}. Because of the orthogonality of Bessel functions \cite{bowman1958introduction}, \eqref{eqn:beforeexpandcos} can be written as 
{\begin{align}
    \bar{\mathcal{C}}(\gamma)&\approx1-\frac{\sqrt{\gamma}}{2^{4M+1}}\sum_{k=0}^{2M}\sum_{l=0}^{2M}\left(\begin{matrix}
        2M\\k
    \end{matrix}\right)\left(\begin{matrix}
        2M\\l
    \end{matrix}\right) \notag\\
    &\quad\quad\times\int_0^\infty J_1(\sqrt{\gamma}\zeta)J_0(d_{k,l}\sqrt{P_\delta T/2}p_0(0){\zeta})d\zeta \notag\\
    &=1.
\end{align}}
\vspace{-1cm}

\section{Proof for Theorem \ref{thm:qpsknogaus}}
\label{app:asymgqpsknogaus}

In this appendix, we will verify the result we obtained in Section \ref{sec:qpsktxsnrccdf} by showing that the asymptotic behavior of average CCDF for FTN signaling using QPSK symbols with $\mathsf{SNR_{tx}}$ fixed does not resemble the average CCDF with Gaussian symbols.

\begin{figure}
    \centering
    \includegraphics[width=1\linewidth]{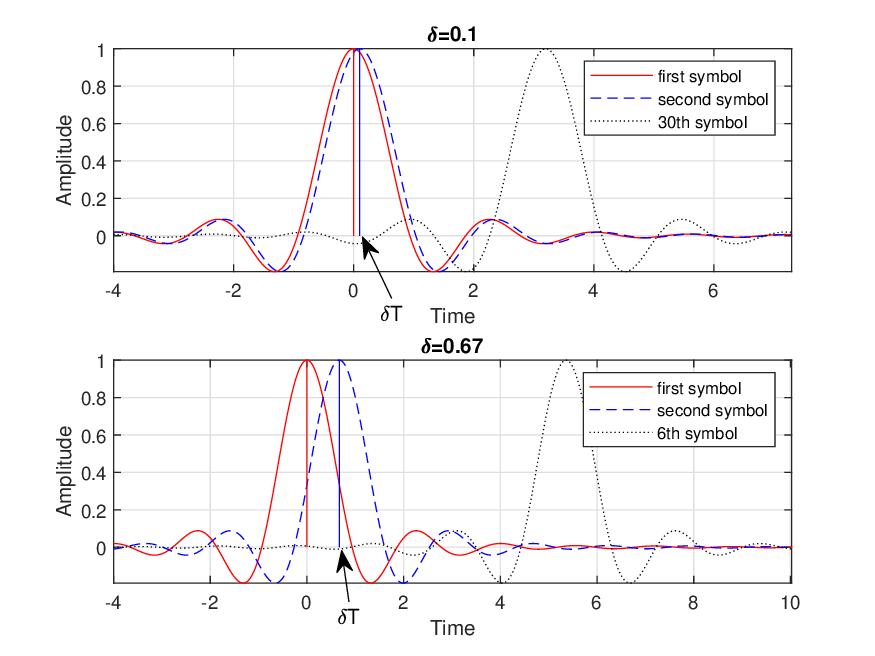}
    \caption{A demonstration of the transmitted FTN signal for two distinct $\delta$ values, showcasing only three symbols as an example.}
    \label{fig:demosmalltau}
\end{figure}

As $\delta$ decreases, transmit pulses modulated with symbols start to pack in a tighter manner. Fig. \ref{fig:demosmalltau} shows a simple demonstration of the signal $x_k(t)$. For smaller $\delta$, the ISI at sampling times becomes more severe. However,  the samples of the process $x(t_0)=\sum_{m=-M}^{M}a_k[m]p(t_0-m\delta T)$ do not converge to Gaussian random variables. 
Despite the presence of a sufficiently large number of symbols, certain symbols within $x(t_0)$ are still multiplied by coefficients significantly smaller than others. 
For instance, as shown in Fig. \ref{fig:demosmalltau}, when $\delta=0.1$ and sampling occurs at $t_0=0$, symbols located far from the sampling time contribute minimally to the sample. Therefore, to determine whether the FTN signal approaches a Gaussian process, it is necessary to examine the Lindeberg condition \cite{lindeberg1922neue}, which is more general than the central limit theorem (CLT). 
\begin{lemma}\cite[Lindeberg condition]{lindeberg1922neue}\label{thm:lindeberg}
    Let $X_l, l=1,\dots, S$ be $S$ mutually independent random variables, each with zero mean and variance $\sigma^2_l$. Let $Z_S=\sum_{l=1}^{S}X_l$ be the sum of $X_l$ with variance   $\sigma^2_{Z,S}=\sum_{l=1}^{S}\sigma_l^2$. If all  the $S$ random variables satisfy   
    \begin{align}
        \underset{S\rightarrow\infty}{\lim}\frac{1}{\sigma^2_{Z,S}}\sum_{l=1}^{S}\mathbb{E}\left[X_l^2\mathbbm{1}_{\left\{|X_l|>\epsilon \sigma_{Z,S}\right\}}\right]=0\label{eqn:lindebergdef}
    \end{align}
    for all $\epsilon>0$, then the variable $\frac{Z_N}{\sigma_{Z,S}}$ converges in distribution to a standard normal random variable as $N\rightarrow\infty$. 
\end{lemma}

The main idea for this theorem is that the CLT can still hold if there is no single or a group of $X_l$'s dominating the variance for the sum variable.  
The ISI we receive mainly comes from neighboring symbols and if we look into the random variable $x_k(t_0)$, we can find that the variance for this random variable $P\delta T\sum_{m=-M}^Mp^2(t_0-m\delta T)$ is dominated by 
the neighboring symbols. Therefore, it is necessary to investigate whether the Lindeberg condition is satisfied for $x_k(t_0)$.
As $\delta$ approaches zero (but not equal to zero), we have the following lemma.

\begin{lemma}
    The variance of random variable $x_k(t_0)$ always converges for arbitrary small but non-zero $\delta$  as $S$ goes to infinity.
\end{lemma}
\begin{proof}
    We observe that the variance $P\delta T\sum_{m=-M}^Mp^2(t_0-m\delta T)$ resembles numerical integration, so when $\delta$ is small enough, we can replace $\delta T$ with $\Delta$, which can be considered as dividing the time period $T$ into small intervals with length $\delta T$. Also, recall that $N=2M+1$. We then have the following 
    \begin{align}
        \underset{M\rightarrow\infty}{\lim}P_k\sum_{m=-M}^Mp^2(t_0-m\delta T)\Delta\approx P_k \int_{-\infty}^{+\infty}p^2(t)dt.
    \end{align}
    Since we assumed that $p(t)$ has unit energy, we conclude that the variance of $x_k(t_0)$ converges to $P_k$ as $\delta$ decreases.
\end{proof}

    We insert $x_k(t)$ into Lindeberg's condition and define  $s^2_N=P_k\delta T\sum_{m=-\frac{N-1}{2}}^{\frac{N-1}{2}}p^2(t_0-m\delta T)$, the left-hand side of \eqref{eqn:lindebergdef} becomes
    \begin{equation}
        \hspace{-0.2cm}\underset{N\rightarrow\infty}{\lim}\frac{1}{s^2_N}\sum_{m=-\frac{N-1}{2}}^{\frac{N-1}{2}}\mathbb{E}\left[|a_k[m]p(t_0-m\delta T)|^2\mathbbm{1}_{\left\{|a_k[m]p(t_0-m\delta T)|>\epsilon s_N\right\}}\right]. \label{eqn:lindeberg}
    \end{equation}
    We know that as $N\rightarrow\infty$, $s^2_N$ converges to $P_k$, therefore we choose $\epsilon<\sqrt{P_k\delta T}p(t_0)$ and $\epsilon$ is  larger than $a_k[m]p(t_0-m\delta T), m\neq 0$. As a result,  \eqref{eqn:lindeberg} becomes 
    \begin{align}
        \frac{\mathbb{E}\left[|a_k[0]p(t_0)|^2\right]}{P_k}=\frac{P_k\delta Tp^2(t_0)}{P_k},
    \end{align}
    which is not zero. Thus Lindeberg's condition is not satisfied. We then conclude that the CLT is not invoked, and the process $x_k(t)$ does not approach the Gaussian process in distribution.

\bibliographystyle{IEEEtran}

\bibliography{main}

\end{document}